\documentclass[10pt,journal,a4paper]{IEEEtran} 

\usepackage{amssymb}
\usepackage{amsmath}
\usepackage{amsmath,bm}
\usepackage{amsthm}
\usepackage{graphicx}
\usepackage{cuted}
\usepackage{subfigure}
\usepackage{multirow}
\usepackage{makecell} 
\usepackage{cite}
\usepackage{enumerate}
\usepackage{enumitem}
\usepackage{hyperref}
\usepackage{stfloats}
\usepackage{footnote}
\usepackage{algorithm}
\usepackage{algorithmic}  
\usepackage{color, soul}
\usepackage[labelsep=period]{caption}

\allowdisplaybreaks[4]

\DeclareMathOperator{\Tr}{Tr}
\DeclareMathOperator{\diag}{diag}

\begin{document}
	\newtheorem{theorem}{\bf~~Theorem}
	\newtheorem{remark}{\bf~~Remark}
	\newtheorem{observation}{\bf~~Observation}
	\newtheorem{definition}{\bf~~Definition}
	\newtheorem{lemma}{\bf~~Lemma}
	\newtheorem{preliminary}{\bf~~Preliminary}
	\newtheorem{proposition}{\bf~~Proposition}
	\newtheorem{comment}{\bf~~Comment}
	\renewcommand\arraystretch{0.9}
	\title{{\Huge Holographic Beamforming for Semantic Communication}}
	\author{\normalsize \IEEEauthorblockN{
			{Shuhao Zeng}, \IEEEmembership{\normalsize Member, IEEE},
			{Haobo Zhang}, \IEEEmembership{\normalsize Member, IEEE},
			{Su Wang}, \IEEEmembership{\normalsize Member, IEEE},\\
			{Boya Di}, \IEEEmembership{\normalsize Senior Member, IEEE},
			{Hongliang Zhang}, \IEEEmembership{\normalsize Member, IEEE},
			{Zhu Han}, \IEEEmembership{\normalsize Fellow, IEEE},\\
			{H. Vincent Poor}, \IEEEmembership{\normalsize Life Fellow, IEEE},
			{and Lingyang Song}, \IEEEmembership{\normalsize Fellow, IEEE}}
		\thanks{Shuhao Zeng is with Department of Electrical and Computer Engineering, Princeton University, NJ, USA, and also with School of Electronic and Computer Engineering, Peking University Shenzhen Graduate School, Shenzhen, China (email: shuhao.zeng96@gmail.com).}
		\thanks{Haobo Zhang is with School of Electronic and Computer Engineering, Peking University Shenzhen Graduate School, Shenzhen, China, and also with Department of Engineering, University of Cambridge, Cambridge, UK (email: haobo.zh97@gmail.com)}
		\thanks{Su Wang is with Department of Electrical and Computer Engineering, Princeton University, NJ, USA (email: hw5731@princeton.edu).}
		\thanks{Boya Di and Hongliang Zhang are with School of Electronics, Peking University, Beijing 100871, China (email: hongliang.zhang@pku.edu.cn; boya.di@pku.edu.cn).}
		\thanks{Zhu Han is with Electrical and Computer Engineering Department, University of Houston, Houston, TX, USA, and also with the Department of Computer Science and Engineering, Kyung Hee University, Seoul, South Korea (email: zhan2@uh.edu).}		
		\thanks{H. Vincent Poor is with Department of Electrical and Computer Engineering, Princeton University, NJ, USA (email: poor@princeton.edu).}
		\thanks{Lingyang Song is with School of Electronic and Computer Engineering, Peking University Shenzhen Graduate School, Shenzhen, China, and also with School of Electronics, Peking University, Beijing 100871, China, and also with Hunan Institute of Advanced Sensing and Information Technology, Xiangtan University, Xiangtan, China (email: lingyang.song@pku.edu.cn).}	
	}
	
	\maketitle
	\begin{abstract}
		Holographic beamforming enabled by metamaterial antennas has been proposed to facilitate spatial multiplexing at low hardware cost and low power consumption. However, existing holographic beamforming schemes are mainly developed for conventional bit-communication systems, which have not considered semantic-level importance and thus cannot be directly applied to support semantic communication. Specifically, in conventional bit communication, all bits are treated as equally important. In contrast, in semantic communication, different semantic information contribute unequally to task completion and therefore has different degrees of importance, with more important information requiring higher transmission quality. Ignoring semantic importance in holographic beamforming causes mismatches between importance of semantic information and its received SNR, thus degrading performances. In this paper, we propose a semantic-importance-aware holographic beamforming scheme enabled by metamaterial antennas with tunable radiated amplitudes to support semantic communication. It is challenging to design semantic-aware holographic beamforming schemes due to non-trivial modeling of the impact of semantic importance and unique amplitude-controlled structures of holographic beamforming. {To address this, we characterize the dependence of semantic communication performance on semantic importance and received SNR via data fitting, and design a semantic-aware holographic beamforming algorithm to ensure reliable delivery of highly important semantic information}. Simulation results validate effectiveness of the proposed method.



	\end{abstract}

\vspace{-.4cm}
		
	\section{Introduction}
	Future sixth generation~(6G) networks are envisioned to enable a variety of emerging services such as autonomous driving and augmented reality. These applications typically entail massive and frequent data exchanges, which poses an unaffordable burden on existing wireless systems constrained by limited spectrum resources. To alleviate the scarcity of bandwidth resources and the surge of data traffic, semantic communication has been listed as one of the promising techniques~\cite{Weng_Semantic_2024}. Unlike conventional bit-communication systems, semantic communication transmits only the information that is essential for achieving communication goals by extracting semantic features from source messages, thereby relieving the communication burden~\cite{Zhao_SecDiff}. 
	
	While semantic communications effectively reduce the amount of information to be transmitted, it alone may not be fully sufficient to handle the dramatic growth of wireless data traffic anticipated in future networks. For example, existing semantic communication methods for textual transmission can reduce the data size by about $80\%$ compared with conventional bit-communication systems~\cite{Liu_explainable_IOT}, whereas reports from the International Telecommunication Union (ITU) indicate that global data traffic is expected to increase nearly tenfold every five years~\cite{Tariq_specular_2020}. These observations motivate the integration of semantic communication with large-scale antenna arrays, whose massive number of antenna elements provide significant spatial multiplexing gains and thus enable more efficient deliveries of semantic information.
	
	
	Nevertheless, conventional phased arrays are not suitable for such large-scale deployments, as they rely on a large number of costly phase shifters and amplifiers, which result in prohibitive hardware costs~\cite{Dai_RIS_2020}. Consisting of numerous metamaterial radiation elements, the emerging reconfigurable holographic surface~(RHS) shows great potential to cope with this issue~\cite{Zeng_HISAC_2025}. Specifically, the RHS utilizes the metamaterial to construct a holographic pattern on its surface based on the holographic interference principle. According to the holographic pattern, each element can control its radiation amplitude by tuning the biased voltages applied to onboard diodes to generate desired directional beams~\cite{Deng_RHS_multi_user_2022}. Thus, without complex phase shifting circuits, the RHS can realize low-cost and energy-efficient dynamic beamforming based on reconfigurable holographic patterns, and such a beamforming technique is also known as holographic beamforming.

	Existing works have proposed various holographic beamforming schemes for application to conventional bit-communication systems~\cite{Deng_RHS_multi_user_2022,Yurduseven_holographic_2020,Wang_RHS_2025,Deng_HDMA_2022}. For example, to direct the beams generated by the RHS towards users, the authors in~\cite{Deng_RHS_multi_user_2022} propose an RHS-based hybrid beamforming scheme with digital beamforming at the base station~(BS) and holographic beamforming at the RHS, which is further optimized via sum-rate maximization. In~\cite{Yurduseven_holographic_2020}, the authors further consider satellite communication scenarios and design a metasurface antenna which performs holographic beamforming to support efficient satellite communication. To reduce the complexity of holographic beamforming design algorithms, the authors in~\cite{Deng_HDMA_2022} introduce holographic-pattern division multiple access~(HDMA), {which constructs a unified holographic pattern on the RHS by superimposing multiple user-specific holographic patterns, each corresponding to a directional beam intended for a particular user}.

	However, existing metamaterial-antenna-aided holographic beamforming schemes developed for conventional bit-communication systems adapt poorly to holographic beamforming designs in semantic communication systems, since they do not take the importance of semantic information into account. Specifically, in conventional bit-communications, source messages are merely mapped into bit sequences, with the objective of guaranteeing the reliable delivery of each bit~\cite{Goldsmith_2005}. Therefore, all transmitted bits are treated as equally important. In contrast, semantic communication extracts and conveys the task-relevant semantic information embedded in source messages~\cite{Getu_Performance_2024}, where the semantic importance of different semantic information to task completions is not uniform. For instance, in a traffic image transmission task, consider an image that includes both vehicles and roadside houses as background. The semantic information related to cars is more important than that related to roadside houses, as the former is essential for traffic monitoring and management. Such differences in semantic importance should be explicitly considered in holographic beamforming designs, ensuring that semantic information of higher importance is delivered with higher SNR. If we directly apply conventional holographic beamforming schemes, the transmission quality of important semantic information can be impaired, thereby degrading semantic communication performance and potentially leading to task failures.
	 
	In this paper, we investigate holographic beamforming schemes to support semantic communication. Specifically, a BS equipped with an RHS semantically transmits images to multiple users. To enable efficient delivery of semantic information, the BS applies digital beamforming, and the RHS performs analog beamforming by tuning the radiated amplitude of each metamaterial element in the RHS to enable spatial multiplexing\footnote{Although prior studies have investigated beamforming designs for semantic MIMO communications~\cite{Weng_Semantic_2024,Wu_Deep_2025,Zhang_Beamforming_2025}, most of them focus on fully digital schemes without incorporating RHS-aided amplitude-controlled analog beamforming, and thus are not applicable in this context. Moreover, fully digital beamforming entails excessive hardware costs and power consumption in large-scale deployments, as it requires a large number of costly and power-hungry RF chains.}. However, it is challenging to design semantic-aware RHS-aided holographic beamforming schemes, since it is non-trivial to mathematically model the relationship between semantic communication performances and RHS configurations, which should explicitly account for semantic importance. Furthermore, unlike conventional phase arrays, the analog beamforming at the RHS operates in an amplitude-controlled manner, and thus the beamforming schemes designed for phased-array-enabled semantic MIMO systems cannot be directly applied. To address these challenges, we advance the state-of-the-art in the following aspects:
	\begin{itemize}
		\item We consider a downlink RHS-aided multi-user communication system, where a BS equipped with an RHS transmits images to multiple users in a semantic manner. To reduce latency, each image is divided into multiple sub-images for parallel transmissions. Semantic information is extracted from each sub-image, and its semantic importance is evaluated. To enable the simultaneous transmissions of multi-stream semantic information to multiple users, we propose a hybrid beamforming scheme, where the BS performs digital beamforming, the RHS performs holographic beamforming, and each user conducts receive digital combining.
		\item We characterize the dependence of semantic communication performance, i.e., overall image reconstruction loss, on semantic importance and received SNR of each sub-image in closed form. To adapt the received SNR of the sub-images to their semantic importance and thereby minimize the image reconstruction loss, we jointly optimize the digital and holographic beamformers to adjust the transmit power, channel gains, and allocations of available spatial channels\footnote{Each MIMO channel between the BS and one user can be decomposed into multiple parallel independent spatial channels~\cite{Lee_ZF_2004}, as further illustrated in Section~\ref{sec_digital_BF}. Each spatial channel carries the semantic information extracted from a single sub-image. In this context, the digital beamformer and combiner, and the holographic beamformer influence the system performance by changing the transmit power, channel gains, and allocations of the equivalent spatial channels.}. Closed-form solutions for the optimal semantic–aware power and channel allocations are derived.
			
		\item Simulation results validate the effectiveness of the proposed semantic-aware holographic beamforming method by comparison with both the holographic beamforming scheme designed for conventional bit-communication systems and the beamforming scheme designed for phased-array-enabled semantic MIMO communication systems.
	\end{itemize}
	
	The rest of this paper is organized as follows. In Section~\ref{sec_system_model}, an RHS-aided semantic communication network is modeled. In Section~\ref{sec_problem_formulation}, a closed-form approximation for semantic communication performance is derived, and a holographic beamforming problem is formulated to improve semantic communication performances by taking the semantic importance into account, which is decomposed into two subproblems. A semantic-importance-aware joint algorithm is proposed to solve the two subproblems alternatively in Section~\ref{sec_algorithm_design}. In Section~\ref{sec_simulation}, simulation results are presented, and finally conclusions are drawn in Section~\ref{sec_conclusion}.

	\section{System Model}
	\label{sec_system_model}
	In this section, we first introduce the scenario for RHS-aided semantic communication. Then, the learning models at the BS and the UEs are introduced for joint source and channel coding and semantic importance acquisition, followed by the presentation of beamforming schemes. 
	\subsection{Scenario Description}
	As shown in Fig.~\ref{sysmodel}, we consider a downlink multi-user semantic communication system, where a base station~(BS) semantically transmits images to $J$ users, enabled by a semantic encoder/decoder and a semantic importance acquisition network\footnote{When considering the other modalities or scenarios, the proposed beamforming framework can remain unchanged. We only need to update two modality-dependent components that serve as inputs to the framework, i.e., (a) revise the definition of semantic importance according to the specific task and dataset, (b) update the mapping between each data stream's semantic transmission performance (which is the image reconstruction loss in our task) and its received SNR by updating the corresponding curve parameters.}. Further, to support multi-user transmissions, the BS is equipped with an RHS to support multi-beam steering, which does not need energy-hungry phase shifters~\cite{Di_holo_2025}, and thus is suitable for large-scale array implementations to provide high directive gain~\cite{Zeng_dual_pol_2024}. Further, for spatial multiplexing, each user is equipped with $M$ antennas, where user $j$ simultaneously receives $K_j$ data streams from the BS. 
	
	\begin{figure*}[!t]
		\centering
		\includegraphics[width=0.9\textwidth]{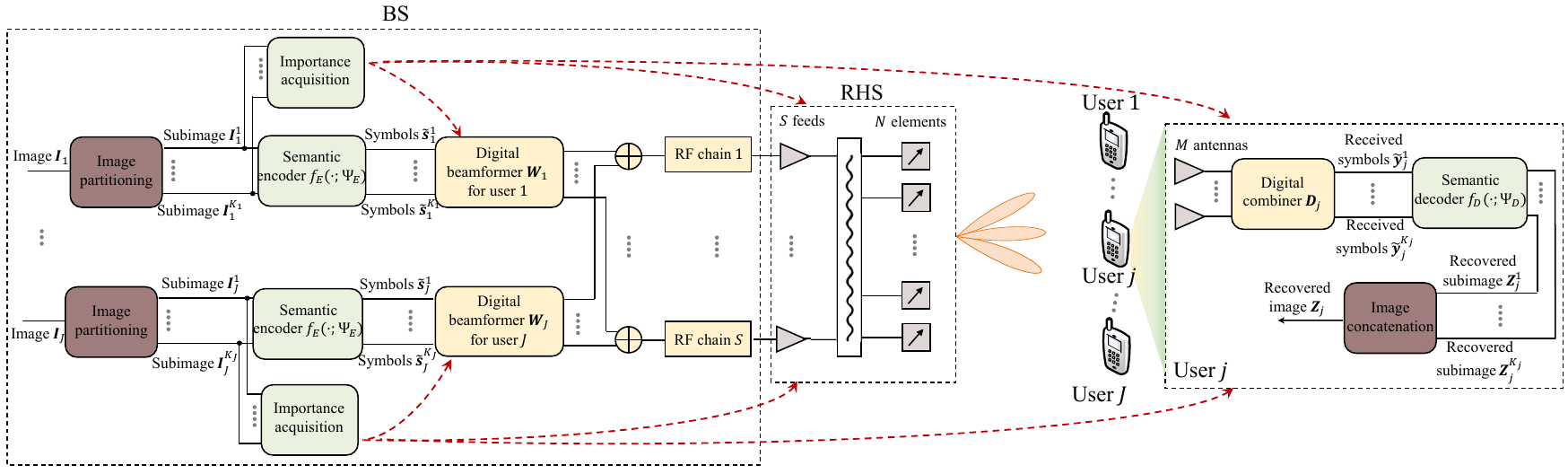}
		\caption{System model of a downlink RHS-aided semantic communication network with multiple users. The red curves in the figure highlight the influence of the importance of semantic information on beamforming\protect\footnotemark.}
		\vspace{-3mm}
		\label{sysmodel}
	\end{figure*}
	
	
	In the following, we summarize the operation of the semantic communication network. To reduce latency, each image for user $j$ is first divided into $K_j$ subimages with equal size. Then, the $K_j$ subimages pass semantic encoders, which extract semantic information from the subimages while performing channel coding to ensure reliable delivery. At the same time, the semantic importance of different subimages are acquired based on semantic segmentation models. Afterwards, by performing hybrid beamforming at the BS and fully digital combining at users, $K_j$ parallel independent equivalent single-input-single-output~(SISO) channels are created for user $j$ without inter-stream and inter-user interference. The transmit symbols, i.e., the output symbols from the semantic encoders, are fed into the $K_j$ parallel channels for simultaneous transmissions. Then, at user $j$, the received $K_j$ symbol streams are fed into the semantic decoder to recover the $K_j$ subfigures, which are then concatenated to reconstruct the original image.
	\footnotetext[2]{Here, the semantic importance is not sent to users or RHS. Instead, it is only utilized for joint beamforming design within the BS.}
	\subsection{Learning Models}
	In this part, we will introduce the semantic encoder/decoder network and the semantic importance acquisition network, respectively.
	
	\subsubsection{Semantic Encoder and Decoder}
	At the BS, an image source first generates images to be transmitted. Let $\bm{I}_j\in \mathbb{R}^{A\times W\times H}$ denote one generated image intended for user $j$, where $A$ is the number of image channels ($A = 1$ for greyscale and $A = 3$ for RGB images), { and $W$ and $H$ represent the number of pixels in the width and height of a figure, respectively.} The image is first divided into $K_j$ smaller sub-images with equal size $W'\times H'$, respectively, for simultaneous transmissions\footnote{The entire image should be spatially partitioned into $K_j$ contiguous subregions, rather than randomly assigning individual pixels into $K_j$ groups. This is because random pixel grouping would disrupt the inherent semantic structure of the original image, thereby preventing the semantic encoder from correctly extracting meaningful features and performing effective compression. 
	
	\par
		
	By contrast, as long as the image is divided into contiguous subregions, the proposed semantic communication framework is robust to the specific partitioning pattern.}. The obtained sub-images are denoted by $\bm{I}_j^{(1)},\dots,\bm{I}_j^{(K_j)}$, respectively. Then, the $K_j$ sub-images pass the semantic encoder, where the semantic information of the figures is extracted and channel coding is finished to generate transmit symbols $\{\widetilde{\bm{s}}_j^{(1)},\dots,\widetilde{\bm{s}}_j^{(K_j)}\}$, i.e.\footnote{The semantic encoder implicitly includes complex-number modulation, i.e., the output signal of the semantic encoder takes complex values.},
	\begin{align}
		\label{semantic_encoder}
		\widetilde{\bm{s}}_j^{(k)}=f_E(\bm{I}_j^{(k)};\bm{\psi}_E),
	\end{align}
	{where $f_E(\cdot;\bm{\psi}_E)$ is the neural-network-based semantic encoder with learnable parameter $\bm{\psi}_E$.} 
	
	When the transmit symbol sequences $\{\widetilde{\bm{s}}_j^{(1)},\dots,\widetilde{\bm{s}}_j^{(K_j)}\}$ transmit through the parallel equivalent SISO channels formed through beamforming, the received symbol sequences at the $j$-th Rx user are denoted by $\{\widetilde{\bm{y}}_j^{(1)},\dots,\widetilde{\bm{y}}_j^{(K_j)}\}$, respectively. The received symbols are first normalized based on the channel gains of the equivalent SISO channels, and then fed into the semantic decoder for image recovery\footnote{The receiver performs more than simple normalization. Specifically, the received RF signals first pass through the RF chains and are converted into baseband signals. Subsequently, digital combining is applied at the receiver. Together with the transmit beamforming, this process aims to suppress inter-user and inter-stream interference. The resulting baseband symbols are then normalized according to the gains of the equivalent SISO channels, i.e., equalization, before being fed into the semantic decoder for channel decoding and semantic recovery.}. Thus, the $k$-th sub-image reconstructed by using the input symbol vector $\widetilde{\bm{y}}_j^{(k)}$ is given by
	\begin{equation}
		\label{sem_decoder}
		\bm{Z}_j^{(k)}=f_D(\widetilde{\bm{y}}_j^{(k)}; \bm{\psi}_D),
	\end{equation}
	{where $f_D(\cdot;\bm{\psi}_D)$ is the neural-network-based semantic decoder with learnable parameters $\bm{\psi}_D$}. Finally, all the $K_j$ sub-images $\{\bm{Z}_j^{(1)},\ldots,\bm{Z}_j^{(K_j)}\}$ are concatenated at user $j$ to reconstruct the original image, which is denoted by $\bm{Z}_j$. 
	
	\subsubsection{Semantic Importance Acquisition}
	In this paper, we focus on the class of objects, an important type of semantic features~\cite{Ven_Three_2022}, and the semantic importance of a subfigure is characterized by the number of pixels $l_j^{(k)}$ within the sub-image whose classes belong to the set $\mathbb{C}$ of desired classes, i.e., $\mathrm{card}\{l_j^{(k)}|l_j^{(k)}\in \mathbb{C}\}$. Here, $\mathrm{card}$ represents the number of elements within a set. A larger number of important pixels indicates higher importance of a sub-image, and thus the corresponding symbols generated by the semantic encoder should be allocated higher transmission quality to guarantee the successful delivery of essential semantic information. {Although the definition of semantic importance is task- and dataset-relevant, the main contribution of this paper is to propose a semantic-importance-aware holographic beamforming framework, which is not tied to any specific definition of semantic importance and does not depend on a particular task\footnote{Within the proposed framework, we first adopt a curve-fitting approach to model a closed-form relationship between each data stream's semantic transmission performance and its received SNR, based on which we can establish an explicit connection between the semantic transmission performance and the beamforming matrix. Here, different semantic-importance levels of the data streams are reflected through different values of the fitted curve parameters. We then formulate a holographic beamforming optimization problem, where we aggregate the semantic transmission performance across all streams into a single objective function, and develop an efficient algorithm to solve it. In this way, semantic importance is explicitly incorporated into holographic beamforming. Such a framework is independent of particular definitions of semantic importance and particular tasks.}.}
	
	To acquire the semantic importance of a sub-figure, a semantic feature extraction model can be utilized to extract the semantic segmentation label for each pixel within the sub-figure. Then, the output semantic segmentation labels for the pixels can be given by
	\begin{align}
		\label{semantic_feature}
		\bm{L}_j^{(k)}=f_S(\bm{I}_j^{(k)};\bm{\psi}_S),
	\end{align}
	{where $f_S(\cdot;\bm{\psi}_S)$ denotes the neural-network-based semantic feature extraction model utilized to acquire the semantic importance of a sub-figure, with $\bm{\psi}_S$ representing model parameters,} $\bm{L}_j^{(k)}\in \mathbb{N}^{W'\times H'}$ is the semantic segmentation labels for subimage $\bm{I}_j^{(k)}$, and the value of the $(w,h)$-th element in $\bm{L}_j^{(k)}$ represents the class to which the $(w,h)$-th pixel in sub-image $\bm{I}_j^{(k)}$ belongs. 
	
	{Based on the semantic segmentation label $\bm{L}_j^{(k)}$, we assign semantic importance to the pixels within the sub-image $\bm{I}_j^{(k)}$. Specifically, we define a semantic weight matrix $\boldsymbol{\Phi}_j^{(k)}$ to characterize the pixel-wise semantic importance, where the $(w,h)$-th element of $\boldsymbol{\Phi}_j^{(k)}$ represents the importance of the $(w,h)$-th pixel in $\bm{I}_j^{(k)}$. A higher weight corresponds to higher semantic importance. According to~\cite{Wu_semantic_2024}, the semantic weight matrix $\boldsymbol{\Phi}_j^{(k)}\in \mathcal{R}^{A\times W'\times H'}$ can be determined by the semantic segmentation label $\bm{L}_j^{(k)}$, i.e.,
	\begin{equation}
		\boldsymbol{\Phi}_j^{(k)}(:,w,h) =
		\begin{cases}
			\theta, & \bm{L}_j^{(k)}(w,h) \in \mathbb{C}, \\[3pt]
			1-\theta, & \bm{L}_j^{(k)}(w,h) \notin \mathbb{C},
		\end{cases}
		\label{eq:weight_matrix_v2}
	\end{equation}
	This indicates that if the class of the $(w,h)$-th pixel belongs to the set $\mathbb{C}$ of desired classes, a high weight $\theta$ is assigned as the semantic weight for all image channels corresponding to this pixel, i.e., $\boldsymbol{\Phi}_j^{(k)}(:,w,h)=\theta$. Otherwise, a smaller weight is assigned, i.e., $\boldsymbol{\Phi}_j^{(k)}(:,w,h)=1-\theta$.}
	

	
	Then, the reconstruction accuracy of the $k$-th sub-image $\bm{I}_j^{(k)}$ at user $j$ can be measured by the average image reconstruction loss weighted by the semantic importance~\cite{Wu_semantic_2024}, i.e.,
	\begin{align}
		\label{weighted_construction_loss_subfigure}
		\xi_j^{(k)}=\mathbb{E}\left\{\frac{1}{W'H'}\|\bm{\Phi}_j^{(k)}\odot (\bm{Z}_j^{(k)}-\bm{I}_j^{(k)})\|_2^2\right\},
	\end{align}
	{where $\bm{Z}_j^{(k)}$ and $\bm{I}_j^{(k)}$ represent the recovered and the original $k$-th sub-image for user $j$, respectively}, $\bm{\Phi}_j^{(k)}$ is the pixel-wise semantic weight of the sub-image $\bm{I}_j^{(k)}$, $\odot$ is Hardmard product, $W'H'$ is the total number of pixels in the sub-figure, $\|\cdot\|_2$ is the Euclidean norm of a vector, {and the expectation $\mathbb{E}\left\{\cdot\right\}$ is taken with respect to the joint distribution of the original sub-figure $\bm{I}_j^{(k)}$ and the recovered sub-figure $\bm{Z}_j^{(k)}$}. To improve the successful transmissions of essential semantic information, the value of one element in semantic weight matrix $\bm{\Phi}_j^{(k)}$ takes a larger value when the corresponding pixel belongs to the desired class $\mathbb{C}$.
	By jointly considering the reconstruction loss of all $K_j$ subfigures $\bm{I}_j^{(1)},\dots,\bm{I}_j^{(K_j)}$, the reconstruction accuracy of the complete figure $\bm{I}_j$ can be given by
	\begin{align}
		\label{weighted_construction_loss}
		\xi_j&=\mathbb{E}\left\{\frac{1}{W'H'K_j}\sum_{k=1}^{K_j}\|\bm{\Phi}_j^{(k)}\odot (\bm{Z}_j^{(k)}-\bm{I}_j^{(k)})\|_2^2\right\}\\
		&=\frac{1}{K_j}\sum_{k=1}^{K_j}\xi_j^{(k)}.\notag	
	\end{align}
	
	
	\subsection{Hybrid Beamforming Models}	
	Define $T$ as the length of the symbol sequence $\widetilde{\bm{s}}_j^{(k)}$ given in (\ref{semantic_encoder}), i.e., $\widetilde{\bm{s}}_j^{(k)}=[s_j^{(k,1)},\dots,s_j^{(k,T)}]^\mathsf{T}$, which is generated by the semantic encoder and corresponds to the $k$-th sub-image intended for user $j$. Consider one encoded symbol $s_j^{(k,t)}$ from the symbol sequence $\widetilde{\bm{s}}_j^{(k)}$. For the simplicity of discussions, we omit the index $t$ in the following. The encoded symbol $s_j^{(k)}$ is assumed to be unit power, i.e., $\mathbb{E}(|s_j^{(k)}|^2)=1$. By concatenating the symbols from $K_j$ different symbol sequences, we can get a encoded symbol vector corresponding to user $j$, i.e., ${\bm{s}}_j=(s_j^{(1)},\dots,s_j^{(K_j)})^T$. Then, the encoded symbol vector for $J$ users can be given by ${\bm{s}} = [{\bm{s}}_1^T,{\bm{s}}_2^T,\ldots,{\bm{s}}_J^T]^T 
	\in \mathbb{C}^{\sum_jK_j \times 1}$. The transmitted semantically encoded symbols ${\bm{s}}$ are first precoded by the digital beamformer $\bm{W}=[\bm{W}_1,\bm{W}_2,\dots,\bm{W}_J]\in \mathbb{C}^{S\times \sum_jK_j}$, where $\bm{W}_j$ is the digital beamformer for the symbol vector $\bm{s}_j$ of user $j$, and $S$ is the number of RF chains at the BS. After digital beamforming, the transmit signals can be given by
	\begin{align}
		\label{tx_signal_digital}
		\bm{s}_D=\bm{W}\bm{s}.
	\end{align} 
	

	
	Then, the BS up-converts the $\bm{s}_D$ over the carrier frequency through $S$ RF chains. Each RF chain is connected to one feed of the RHS, i.e., the number of RF chains $S$ is equal to the number of feeds at the RHS. Further, each feed excites one row of $N_{sub}$ elements, and thus the RHS is a uniform planar array~(UPA) with $N = N_{sub} \times S$ radiation elements. The feeds transform the high frequency current into an electromagnetic wave, i.e., reference wave, propagating on the RHS~\cite{Boya_RHS_online}. The reference wave will excite RHS elements one by one to generate transmitted signals~\cite{Yue_Hybrid_2025}, where the transmitted signals $\bm{x}\in \mathbb{C}^{N\times 1}$ of the RHS elements can be expressed as
	\begin{align}
		\label{tx_signal_RHS}
		\bm{s}_T=\bm{A}\bm{Q}\bm{s}_D.
	\end{align}
	Here, $\bm{A}=\diag\{a_{1},\dots,a_{N}\}\in\mathbb{C}^{N\times N}$ is diagonal holographic beamforming matrix consisting of the radiation amplitudes of different RHS elements, and $\bm{Q}\in \mathbb{C}^{N\times S}$ collects the phase shifts and amplitude changes from the feeds to the RHS elements. Here, the radiation amplitude $a_{n}$ takes value within $[a_{min},a_{max}]$. Further, the matrix $\bm{Q}$ can be expressed as $\bm{Q}=\diag\{\bm{q}_1,\dots,\bm{q}_S\}$, {where $\bm{q}_s$ is a row vector of length $N_{sub}$, i.e., $\bm{q}_s=[q_{s,1},\dots,q_{s,n},\dots,q_{s,N_{sub}}]$}. The $n$-th element $q_{s,n}$ describes the phase shift and amplitude change from the $s$-th feed to the $n$-th RHS element in the $s$-th row, which can be expressed as $q_{n,s}=\sqrt{(1-p_{on}a_{max}^2)^{n-1}}e^{-jk_Rd_{n,s}}$.
	Here, $p_{on}$ is the activation probability of the RHS elements, $a_{max}$ is the maximum radiation amplitude of the RHS elements, $k_R$ is the wavenumber corresponding to the reference wave, and $d_{n,s}$ is the distance between the $n$-th RHS element in the $s$-th row and the $s$-th feed.
	
	After propagation through wireless channels, the transmitted signals of the RHS can be received by different users. Then, user $j$ will perform combining over received signals from $M$ antennas by utilizing a $K_j\times M$ digital combiner $\bm{D}_j$, where $K_j$ and $M$ are the number of data streams and the number of antennas. Based on (\ref{tx_signal_digital}) and (\ref{tx_signal_RHS}), the combined signal at user $j$ can be expressed as
	\begin{align}
		\label{rx_signal_combine}
		\bm{y}_j=&\bm{D}_j\bm{H}_j\bm{A}\bm{Q}\bm{W}\bm{s}+\bm{D}_j\bm{n}_j\notag\\
		=&\bm{D}_j\bm{H}_j\bm{A}\bm{Q}\bm{W}_j\bm{s}_j+\sum_{j'\neq j}\bm{D}_j\bm{H}_j\bm{A}\bm{Q}\bm{W}_{j'}\bm{s}_{j'} +\bm{D}_j\bm{n}_j,\notag\\
		=&\bm{E}_j^{\text{diag}}\bm{s}_j+(\bm{E}_j-\bm{E}_j^{\text{diag}})\bm{s}_j+\sum_{j'\neq j}\bm{D}_j\bm{H}_j\bm{A}\bm{Q}\bm{W}_{j'}\bm{s}_{j'}\notag\\ &+\bm{D}_j\bm{n}_j,
	\end{align}
	where $\bm{H}_j\in \mathbb{C}^{M\times N}$ is the channel matrix from the $N$ RHS elements to the $M$ antennas of user $j$, {$\bm{A}$ and $\bm{Q}$ are the diagonal holographic beamforming matrix and the propagation matrix from the feeds to the RHS elements defined in (\ref{tx_signal_RHS}), respectively, $\bm{W}=[\bm{W}_1,\bm{W}_2,\dots,\bm{W}_J]$ is digital beamforming matrix at the BS defined in (\ref{tx_signal_digital})}, and $\bm{n}_j$ is an additive white Gaussian noise vector at user $j$, each with zero mean and variance $\sigma^2$ representing receiver noise. Furthermore, $\bm{E}_j=\bm{D}_j\bm{H}_j\bm{A}\bm{Q}\bm{W}_j$ is the equivalent channel from the BS to user $j$, and $\bm{E}_j^{diag}$ represents the diagonal matrix consisting of the diagonal elements of matrix $\bm{E}_j$. In (\ref{rx_signal_combine}), the first term is the desired signal, the second term represents the inter-stream interference, the third term is inter-user interference, and the last term denotes noise. 
	
	In addition, the received signal vector can be expanded as $\bm{y}_j=[y_j^{(1)},\dots,y_j^{(K_j)}]^\mathsf{T}$, where $y_j^{(k)}$ is the received symbol corresponding to the transmit symbol $s_j^{(k)}$. By collecting and concatenating received symbols $y_j^{(k)}$ over $T$ symbol periods, we can derive the received symbol sequence $\widetilde{\bm{y}}_j^{(k)}=[y_j^{(k,1)},\dots,y_j^{(k,T)}]^{\mathsf{T}}$, which will be further processed by the semantic decoder to recover the $k$-th subimage $\bm{I}_j^{(k)}$, as indicated in (\ref{sem_decoder}).
	
	
	To guarantee the transmission quality, we assume that both the inter-stream interference and the inter-user interference are eliminated through hybrid beamforming and combining\footnote{We would like to point out that it is reasonable to assume perfect interference suppression, which is achieved primarily by the digital beamforming and receive combining~\cite{Lee_ZF_2004}. Specifically, within the alternating optimization framework, when optimizing the amplitudes of the RHS elements, we do not explicitly impose the inter-user and inter-stream interference suppression constraints, since the main objective of the RHS optimization is to enhance the effective channel gains of different data streams through RHS configuration design. Therefore, the no-interference requirement is not the main burden of the RHS optimization. Instead, after each RHS update, we design the digital beamformer and combiner based on the block-diagonalization~(BD)-based framework in Section IV-A, which can completely eliminate both inter-user and inter-stream interference~\cite{Lee_ZF_2004}. This is because as shown in Section IV-A, the digital precoder is designed in the null space of the effective multi-user channel to remove inter-user interference, and the remaining inter-stream interference is further eliminated by the SVD-based transmit/receive design. Therefore, the beamforming solution obtained in each iteration, as well as the final converged solution, satisfies the perfect interference suppression requirement, indicating that the proposed algorithm is robust and the constraints in (\ref{no_inter_stream_interference}) and (\ref{no_inter_user_interference}) do not cause feasibility or convergence issues.}, i.e.,
	\begin{align}
		\label{no_inter_stream_interference}
		\bm{E}_j-\bm{E}_j^{\text{diag}}=\bm{0},~\forall j
	\end{align}
	\begin{align}
		\label{no_inter_user_interference}
		\bm{H}_j\bm{A}\bm{Q}\bm{W}_{j'}=\bm{0},~\forall j'\neq j.
	\end{align}
	Therefore, the received signal in (\ref{rx_signal_combine}) can be further written as
	\begin{align}
		\label{rx_signal_combine_v3}
		\bm{y}_j=\bm{E}_j^{\text{diag}}\bm{s}_j+\bm{D}_j\bm{n}_j.
	\end{align}
	
	Define $\bm{E}_j(k,k)$ as the $k$-th diagonal element in the equivalent channel matrix $\bm{E}_j$ (and thus also the $k$-th diagonal element in matrix $\bm{E}_j^{\text{diag}}$), and use $\bm{d}_j^{(k)}$ to represent the $k$-th row of the receive combining matrix $\bm{D}_j$. Then, (\ref{rx_signal_combine_v3}) indicates that we formulate $K_j$ independent parallel equivalent SISO channels between the BS and user $j$, with the channel gain and noise of the $k$-th equivalent channel given by $\bm{E}_j(k,k)$ and $\bm{d}_j^{(k)}\bm{n}$, respectively. Therefore, we can derive the received signal-to-noise-ratio~(SNR) $\gamma_j^{(k)}$ corresponding to the $k$-th subimage of user $j$ as follows:
	\begin{align}
		\label{SNR}
		\gamma_j^{(k)}=\frac{|\bm{E}_j(k,k)|^2}{\|\bm{d}_j^{(k)}\|_2^2\sigma^2},
	\end{align}
	The received SNR $\gamma_j^{(k)}$ can have an influence on the reconstruction quality $\xi_j^{(k)}$ of the $k$-th subimage given in (\ref{weighted_construction_loss_subfigure}), i.e.,
	\begin{align}
		\xi_j^{(k)}=f_j^{(k)}(\gamma_j^{(k)}),
	\end{align}
	where the function $f_j^{(k)}(\cdot)$ is dependent on the semantic importance of the $k$-th data stream, and will be approximated in closed form through data fitting in the following section. 
	
	\section{Holographic Beamforming Problem Formulation for Semantic Communications}
	\label{sec_problem_formulation}
	In this section, we first derive a closed-form approximation for the image reconstruction loss by using data regression, based on which a holographic beamforming problem is formulated and then decomposed for efficient solutions.
	\subsection{Approximations of Image Reconstruction Loss}
	In the context of image transmission, the performance of semantic communication is measured by the weighted image reconstruction loss in (\ref{weighted_construction_loss_subfigure}) and (\ref{weighted_construction_loss}). However, semantic communication highly relies on neural networks for semantic extraction and recovery, the black-box nature of which hinders theoretical analysis, making the semantic performance, i.e., the weighted reconstruction loss, unable to be acquired precisely. A commonly adopted method for tackling this problem is data regression~\cite{Yan_Resource_2022}, which obtains the mapping from SNR $\gamma_j^{(k)}$ to $\xi_j$ through sufficient experimental instances and curve fitting, as shown below.
	
	\begin{figure}[!t]
		\centering
		\includegraphics[width=0.38\textwidth]{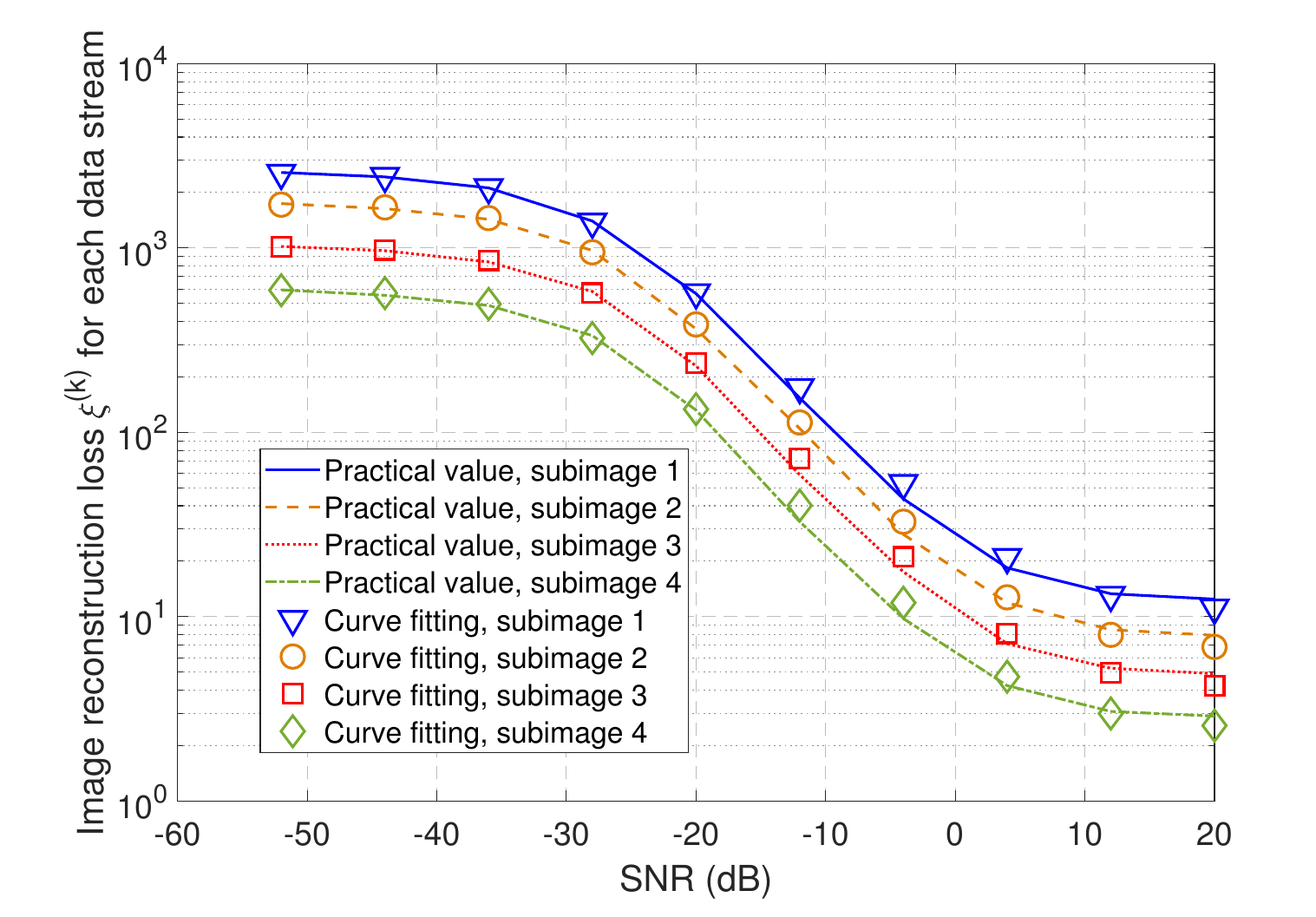}
		\caption{Sub-image reconstruction loss versus SNR. The experiment is conducted on the validation set of the Pascal VOC 2012 dataset. Each image is divided into $K=4$ sub-images, which are then ranked according to their semantic importance. Sub-images with the same rank across different images are grouped together, yielding $K=4$ groups of subimages with distinct average importance levels. Among them, sub-image group $1$ has the highest semantic importance, while group $4$ has the lowest.}
				\vspace{-3mm}
		\label{reconstruction_loss_vs_SNR}
	\end{figure}
	
	Fig.~\ref{reconstruction_loss_vs_SNR} shows the influence of the receive SNR on the image reconstruction loss $\xi_j^{(k)}$ acquired through experiments\footnote{In Fig.~\ref{reconstruction_loss_vs_SNR}, we plot the SNR range of $<-10$~dB in order to demonstrate the overall changing trend of the semantic communication performance with respect to SNR, i.e., the $S$-shape, and validates the selection of fitting function as the logistic function. In addition, for the benchmark algorithms such as the random beamforming algorithm, the receive SNR can drop below $-10$~dB, and thus in order to compare the performance of the proposed algorithm with that of existing methods, we should model the performance when the SNR is below $-10$~dB.}. Here, we acquire the image reconstruction loss via the image dataset Pascal-VOC 2012, which has more than $1000$ images~\cite{Everingham_PASCAL_2012,Quan_Siamese_2024}. Specifically, we divide each image into four sub-images and sort them based on their semantic importance. By grouping together the sub-images that share the same rank across different images, we obtain four sets of sub-images, each characterized by a distinct average importance. For each stream, the contained sub-images are sequentially fed into the semantic encoder\footnote{The training detail of the semantic encoder and decoder will be introduced in the simulation part.}, and the average reconstruction loss $\widetilde{\xi}_j^{(k)}$ over different subimages is measured to approximate the image reconstruction loss $\xi_j^{(k)}$ in (\ref{weighted_construction_loss_subfigure}). In this setup, one single SISO channel is considered between the semantic encoder and decoder, where additive white Gaussian noise is introduced. 
	
	As shown in Fig.~\ref{reconstruction_loss_vs_SNR}, for each data stream, $\xi_j^{(k)}$ follows an $S$-shape with respect to $\gamma_j^{(k)}$, which is also revealed in~\cite{Zhang_Beamforming_2025}. Further, such $S$-shape curve also holds in the text transmission task~\cite{Mu_Het_2023,Getu_Performance_2024}. Therefore, similar to \cite{Mu_Het_2023}, the generalized logistic function could be utilized to well approximate $\xi_j^{(k)}(\gamma_j^{(k)})$, as follows,
	\begin{align}
		\label{fitting_curve}
		\xi_j^{(k)}= f_j^{(k)}(\gamma_j^{(k)})\approx \frac{b_j^{(k)}}{c_j^{(k)}+(\gamma_j^{(k)})^{e_j^{(k)}}}+d_j^{(k)},
	\end{align}
	where $d_j^{(k)}$, $b_j^{(k)}$, $c_j^{(k)}$, and $e_j^{(k)}$ are positive constants obtained through curve fitting\footnote{These fitting parameters are dynamically updated online based on the statistics of the transmitted data.}. From Fig.~\ref{reconstruction_loss_vs_SNR}, we can find that the proposed logistic function in (\ref{fitting_curve}) can accurately characterize the influence of SNR on the image reconstruction loss.
	
	{In addition, by observing Fig.~\ref{reconstruction_loss_vs_SNR}, we can have the following two remarks are in order.}
	\begin{remark}
As the received SNR increases, the weighted reconstruction loss decreases due to the reduced impact of noise on transmission. When the SNR is high, however, the loss saturates since errors from the channel become negligible and the residual distortion mainly stems from semantic encoding and decoding, {which cannot be completely eliminated due to the high compression rate at the semantic encoder~\cite{Thomas_Info_Theory}}. Similarly, in the low-SNR regime, the loss also saturates because both original and reconstructed pixel values are confined to the range $[0, 255]$, thereby bounding the maximum distortion.
	\end{remark}
	
\begin{remark}
	 The subimage groups with higher semantic importance have greater influences on the image reconstruction loss, which motivates us to ensure better transmission quality for sub-image groups with higher semantic importance.
\end{remark}
	
	From Fig.~\ref{reconstruction_loss_vs_SNR}, we can find that the impact of the SNR $\gamma_j^{(k)}$ on the average reconstruction loss $\widetilde{\xi}_j^{(k)}$ (and thus $\xi_j^{(k)}$) differs across data streams, which is due to different levels of semantic importance. Such differences can be fully captured by the different values of the constants $d_j^{(k)}$, $b_j^{(k)}$, $c_j^{(k)}$, and $e_j^{(k)}$ across data streams. Then, we have the following remark:
	\begin{remark}
		The proposed model explicitly accounts for variations in semantic importance across different sub-images by incorporating the semantic importance-dependent parameters, namely $d_j^{(k)}$, $b_j^{(k)}$, $c_j^{(k)}$, and $e_j^{(k)}$, into the image reconstruction loss expression $\xi_j^{(k)}$ in (\ref{fitting_curve}), which further facilitates semantic-importance-aware holographic beamforming problem formulations and algorithm designs, as presented below.
	\end{remark}
	
	\vspace{-.3cm}
	\subsection{Problem Formulation}
	We formulate a reconstruction loss minimization problem by jointly optimizing the digital beamformer $\{\bm{W}_j\}_j$, the holographic beamformer $\bm{A}$ at the RHS (i.e., the radiation amplitudes of different RHS elements), and digital combiner $\{\bm{D}_j\}_j$ at the users, i.e.,
	\begin{subequations}\label{opt_problem}
		\begin{align}
			\label{obj}
			&\min_{\{\bm{W}_j\}_j,\bm{A},\{\bm{D}_j\}_j} \sum_j \xi_j=\sum_j\frac{1}{K_j}\sum_{k=1}^{K_j}\xi_j^{(k)},\\
			\label{cons_max_sum_power}
			s.t.~&\sum_j\Tr(\bm{W}_j\bm{W}_j^{H})\le P,\\
			\label{cons_max_radiated_power}
			&\Tr(\bm{A}\bm{Q}\bm{W}\bm{W}^{H}\bm{Q}^{H}\bm{A}^{H})\le P,\\
			\label{cons_discrete_polarizability}
			&a_{n}\in [a_{min},a_{max}],\\
			&(\ref{no_inter_stream_interference}),(\ref{no_inter_user_interference}).	\notag	
		\end{align}
	\end{subequations}
	Constraint (\ref{cons_max_sum_power}) indicates that the output power of the RF chains at the BS cannot extends $P$. Constraint (\ref{cons_max_radiated_power}) is leakage power constraint, i.e., the power that leaks out of the RHS surface and radiates towards free space cannot exceed the overall transmit power $P$. Constraint (\ref{cons_discrete_polarizability}) implies the range of values for the radiation amplitudes of the RHS elements. Constraints (\ref{no_inter_stream_interference}) and (\ref{no_inter_user_interference}) requires that the inter-stream interference and inter-user interference should be eliminated to guarantee transmission qualities of semantic information. Based on (\ref{opt_problem}), we can find that the digital beamformer $\{\bm{W}_j\}_j$, holographic beamformer $\bm{A}$, and digital combiner $\{\bm{D}_j\}_j$ are coupled, and thus it is non-trivial to solve the hybrid beamforming problem in (\ref{opt_problem}). To cope with this issue, we decompose the hybrid beamforming problem in (\ref{opt_problem}) into two subproblems, as shown in the following.
	\vspace{-.3cm}
	\subsection{Problem Decomposition}
	To efficiently solve problem (\ref{opt_problem}), we decompose it into the following two subproblems, i.e., 
	\subsubsection{Digital Beamforming and Combining} 
	Given the holographic beamforming matrix $\bm{A}$ of the RHS, the digital beamforming and combining subproblem\footnote{We would like to point out that here we do not decompose the digital beamformer $\{\bm{W}_j\}_j$ and the digital combiner $\{\bm{D}_j\}_j$ for separate design. This is because in this case, for the digital beamforming subproblem, we need to cancel both inter-user and inter-stream interference by only designing the digital beamformer $\{\bm{W}_j\}_j$, which can lead to lower receive SNR when compared with the joint design of the digital beamformer and combiner~\cite{Lee_ZF_2004}. This thus motivates us to formulate the joint digital beamformer and combiner design problem as given in (\ref{opt_problem_DBF}).} can be given by
	\begin{subequations}\label{opt_problem_DBF}
		\begin{align}
			\label{obj_DBF}
			&\min_{\{\bm{W}_j\}_j,\{\bm{D}_j\}_j} \sum_j\frac{1}{K_j}\sum_{k=1}^{K_j}\xi_j^{(k)},\\
			\label{cons_max_sum_power_DBF}
			s.t.~&\sum_j\Tr(\bm{W}_j\bm{W}_j^{H})\le P,\\
			\label{cons_max_radiated_power_DBF}
			&\Tr(\bm{A}\bm{Q}\bm{W}\bm{W}^{H}\bm{Q}^{H}\bm{A}^{H})\le P,\\
			&(\ref{no_inter_stream_interference}),(\ref{no_inter_user_interference}).	\notag	
		\end{align}
	\end{subequations}
	
	\subsubsection{Holographic Beamforming}
	Given the digital beamformer $\{\bm{W}_j\}_j$ at the BS and receive combiner $\{\bm{D}_j\}_j$ at users, the holographic beamforming subproblem can be given by
	\begin{subequations}\label{opt_problem_HBF}
		\begin{align}
			\label{obj_HBF}
			&\min_{\bm{A}} \sum_j\frac{1}{K_j}\sum_{k=1}^{K_j}\xi_j^{(k)},\\
			\label{cons_max_radiated_power_HBF}
			s.t.~&\Tr(\bm{A}\bm{Q}\bm{W}\bm{W}^{H}\bm{Q}^{H}\bm{A}^{H})\le P,\\
			\label{cons_discrete_polarizability_HBF}
			&a_{n}\in [a_{min},a_{max}],\\
			&(\ref{no_inter_stream_interference}),(\ref{no_inter_user_interference}).\notag	
		\end{align}
	\end{subequations}
	
	\section{Semantic-Aware Beamforming Algorithm Design}
	\label{sec_algorithm_design}
	In this section, we design algorithms to solve the two subproblems (\ref{opt_problem_DBF}) and (\ref{opt_problem_HBF}), respectively, based on which an overall semantic importance aware algorithm is proposed to solve the reconstruction loss minimization problem in (\ref{opt_problem}).
	\subsection{Digital Beamforming Design}
	\label{sec_digital_BF}
	As indicated in (\ref{no_inter_stream_interference}) and (\ref{no_inter_user_interference}), we aim to mitigate inter-user interference and inter-stream interference so as to improve the transmission quality of semantic information. To achieve this, we adopt the block diagonalization-based digital beamforming and combining framework\footnote{The core contribution of this paper lies in the introduction of a novel semantic-importance-aware beamforming framework, rather than in the development of a specific digital beamforming algorithm. Moreover, this framework remains applicable regardless of whether the digital beamforming is implemented using the BD scheme or alternative schemes such as the regularized zero-forcing~(RZF)~\cite{Peel_RZF_2005}.}~\cite{Lee_ZF_2004}, where the inter-user interference is eliminated by casting the digital beamformer into the null-space of the channel matrix while the inter-stream interference is mitigated through singular value decomposition~(SVD) beamforming. Specifically, the digital beamformer $\{\bm{W}_j\}_j$ and combiner $\{\bm{D}_j\}_j$ that can mitigate the inter-user and inter-stream interference can be given by~\cite{Lee_ZF_2004}
	\begin{align}
		\label{digital_BF}
		\bm{W}_j=\widetilde{\bm{V}}_j^{(0)}\bm{V}_j^{(1)}\bm{O}_j\bm{P}_j^{\frac{1}{2}},\\
		\label{digital_CB}
		\bm{D}_j=\bm{O}_j^T(\bm{U}_j^{(1)})^H.
	\end{align}
	Here, $\widetilde{\bm{V}}_j^{(0)}$, $\bm{V}_j^{(1)}$, and $\bm{U}_j^{(1)}$ are constant matrices determined by the channel matrix $\bm{H}$, the radiation amplitudes $\bm{A}$ of the RHS, and the influence $\bm{Q}$ of the reference wave, which can effectively remove inter-user and inter-stream interference. The definitions of these matrices can be found in Appendix~\ref{app_def_db_dc}.Further, $\bm{P}_j=\diag\{p_j^{(1)},\dots,p_j^{(K_j)}\}$ is the power allocation matrix, and $\bm{O}_j\in \mathbb{R}^{K_j\times K_j}$ is equivalent channel allocation matrix, where the $(k,k')$-th element of matrix $\bm{O}_j(k,k')=1$ if the $k$-th equivalent channel is allocated to transmit the $k'$-th subimage. Otherwise, $\bm{O}_j(k,k')=0$. 
	\begin{remark}
		By utilizing the digital beamforming and combining method given in (\ref{digital_BF}) and (\ref{digital_CB}), $K_j$ equivalent independent parallel single-input-single-output channels are constructed between the BS and user $j$ without inter-stream and inter-user interference. 
		
		The channel gain corresponding to the $k$-th equivalent channel is equal to $\sigma_j^{(k)}$, where the definition of $\sigma_j^{(k)}$ can be found in Appendix~\ref{app_sub_SVD}, while the noise power is $\sigma^2$. In addition, the matrix $\bm{P}_j$ describes the transmit power allocation among the $K_j$ spatial channels and $\bm{O}_j$ indicates how the $K_j$ channels are allocated for the delivery of different sub-images.
	\end{remark}

	Based on (\ref{digital_BF}) and (\ref{digital_CB}), the digital beamforming and combining problem in (\ref{opt_problem_DBF}) can be transformed into the following equivalent channel $\{\bm{O}_j\}_j$ and power allocation $\{\bm{P}_j\}_j$ problem, i.e.,
	\begin{subequations}\label{opt_problem_DBF_v2}
		\begin{align}
			\label{obj_DBF_v2}
			&\min_{\{\bm{P}_j\}_j,\{\bm{O}_j\}_j} \sum_j\frac{1}{K_j}\sum_{k=1}^{K_j}\xi_j^{(k)},\\
			\label{cons_max_sum_power_DBF_v2}
			s.t.~&\sum_j\sum_{k=1}^{K_j}p_j^{(k)}\le P,\\
			\label{cons_binary}
			& \bm{O}_j(k,k')\in\{0,1\},\forall j,k,k'\\
			\label{cons_orthogonal_channel_allocation}
			& \sum_{k=1}^{K_j} \bm{O}_j(k,k')=1,~\sum_{k'=1}^{K_j} \bm{O}_j(k,k')=1, \forall j,k,k',
		\end{align}
	\end{subequations}
	Here, constraint (\ref{cons_binary}) indicates that each element of the channel allocation matrix $\bm{O}_j$ can only take binary values, where $\bm{O}_j(k,k')=1$ indicates that the $k$-th equivalent channel is allocated for the transmission of the $k'$-th subimage, and $\bm{O}_j(k,k')=0$, otherwise. Constraint (\ref{cons_orthogonal_channel_allocation}) shows that each subchannel can only be allocated to one data stream\footnote{We would like to point out that constraint~(\ref{cons_max_radiated_power_DBF}) can be safely omitted for simplicity, since any feasible solution to problem~(\ref{opt_problem_DBF_v2}) already satisfies it. Specifically, constraint~(\ref{cons_max_radiated_power_DBF}) can be rewritten as $\sum_{s=1}^S (\bm{l}_s^H \bm{l}_s)(\bm{W}[s,:]\bm{W}^H[s,:]) \le P$, where $\bm{l}_s$ denotes the $s$-th column of the matrix $\bm{A}\bm{Q}$, and $\bm{W}[s,:]$ represents the $s$-th row of the digital beamformer $\bm{W}$.  Note that $\bm{l}_s^H\bm{l}_s$ quantifies the ratio between the radiated power and the input power corresponding to the $s$-th RHS element driven by the $s$-th feed, which satisfies $\bm{l}_s^H\bm{l}_s \le 1$. Therefore, according to constraint~(\ref{cons_max_sum_power_DBF_v2}), the leakage power constraint~(\ref{cons_max_radiated_power_DBF}) is satisfied.}.
	
	We first consider the design of the equivalent channel allocation matrix $\{\bm{O}_j\}_j$. Without loss of generality, we assume that the channel gains $\sigma_j^{(k)}$ of the $K_j$ equivalent channels for user $j$ satisfy, i.e., $\sigma_j^{(1)}\ge\sigma_j^{(2)}\cdots\ge\sigma_j^{(K_j)}$. Then, we have the following theorem.
	\begin{theorem}
		\label{the_optimal_channel_allocation}
		(Optimal channel allocation) To reduce the image reconstruction loss, the equivalent channels with higher channel gain should be allocated for the transmissions of subimages with higher average semantic importance,~i.e.,
\begin{align}
	\label{optimal_channel_allocation}
	\bm{O}_j^*(k_1,k_2)=\left\{
	\begin{aligned}
		&1,~k_1=\mathcal{I}(k_2),\\
		&0,~\text{otherwise}.
	\end{aligned}
	\right.
\end{align}
where $\mathcal{I}(k_2)$ represents denotes the ranking position of the average semantic importance of the $k_2$-th subimage among all $K_j$ subimages of user $j$. 
	\end{theorem}
	\begin{proof}
		See Appendix~\ref{app_optimal_channel_allocation}.
	\end{proof}
	
	Then, given the optimal equivalent channel allocations in (\ref{optimal_channel_allocation}), the power allocation problem can be given by
	\begin{subequations}\label{opt_problem_DBF_v3}
		\begin{align}
			\label{obj_DBF_v3}
			&\min_{\{p_j^{(k)}\}} \sum_j\frac{1}{K_j}\sum_{k=1}^{K_j}\xi_j^{(k)},\\
			\label{cons_max_sum_power_DBF_v3}
			s.t.~&\sum_j\sum_{k=1}^{K_j}p_j^{(k)}\le P.
		\end{align}
	\end{subequations}
By substituting the optimal subchannel allocation in (\ref{optimal_channel_allocation}), (\ref{digital_BF}), (\ref{digital_CB}), and $\bm{E}_j=\bm{D}_j\bm{H}_j\bm{A}\bm{Q}\bm{W}_j$ into (\ref{SNR}), the received SNR corresponding to the $k$-th subimage can be rewritten as
	\begin{align}
		\label{SNR_v2}
		\gamma_j^{(k)}=\frac{(\sigma_j^{(\mathcal{I}(k))})^2p_j^{(k)}}{\sigma^2}.
	\end{align}
	By substituting (\ref{SNR_v2}) into (\ref{fitting_curve}), the objective function (\ref{obj_DBF_v3}) of the power allocation problem can be rewritten as
	\begin{align}
		\sum_j\frac{1}{K_j}\!\sum_{k=1}^{K_j}\xi_j^{(k)}\!=\!\sum_{j,k}\frac{1}{K_j}\!\left(\frac{b_j^{(k)}}{c_j^{(k)}\!+\!(\frac{(\sigma_j^{(\mathcal{I}(k))})^2p_j^{(k)}}{\sigma^2})^{e_j^{(k)}}}\!+\!d_j^{(k)}\right).
	\end{align}	
	The minimization problem in (\ref{opt_problem_DBF_v3}) is a constrained optimization problem, which can be efficiently solved through the interior point method~\cite{Boyd_convex}. Specifically, when $e_j^{(k)}<1,\forall j,k$, the power allocation problem in (\ref{opt_problem_DBF_v3}) is convex and thus the interior point method is guaranteed to converge to the global optimal solution. The overall digital beamforming and combining algorithm is summarized in Algorithm~\ref{algorithm_digital_BF_CB}
	\begin{algorithm}[!tpb]
		\caption{Digital beamforming and combining design}
		\label{algorithm_digital_BF_CB} 
		\begin{algorithmic}[1]
			\REQUIRE RHS radiation amplitudes $\bm{A}$
			\STATE Derive the optimal channel allocation $\{\bm{O}_j^*\}$ by (\ref{optimal_channel_allocation})
			\STATE Derive the optimal power allocation $\{(p_{j}^{(k)})^*\}$ by solving problem (\ref{opt_problem_DBF_v3}) through the interior method
			\STATE Derive the optimal digital beamformer $\{\bm{W}_j^*\}$ and combiner $\{\bm{D}_j^*\}$ by (\ref{digital_BF}) and (\ref{digital_CB})
			\ENSURE{The optimal digital beamformer $\{\bm{W}_j^*\}$ and combiner $\{\bm{D}_j^*\}$} 
		\end{algorithmic}
	\end{algorithm}
	
	Moreover, in the low SNR regime, we can even get a closed form expression for the optimal solution to this problem, as shown in the following theorem.
{
	\begin{theorem}
		\label{the_optimal_power_allocation}
		(Optimal power allocation) 
		In the high SNR regime $\gamma_j^{(k)} \gg (c_j^{(k)})^{\frac{1}{e_j^{(k)}}}$ and when the fitting parameter $e_j^{(k)}<1$, the optimal solution to the power allocation problem in (\ref{opt_problem_DBF_v3}) is given by
		\begin{align}
			\label{power_allocation}
		(p_j^{(k)})^* = \left( \frac{e_j^{(k)} b_j^{(k)}}{\mu} \right)^{\frac{1}{e_j^{(k)} + 1}} \left( \frac{\sigma^2}{(\sigma_j^{(o(k))})^2} \right)^{\frac{e_j^{(k)}}{e_j^{(k)} + 1}},
		\end{align}
		Here, $\mu$ is Lagrange multiplier, which is selected such that the overall transmit power is equal to the maximum allowed power, i.e., $\sum_j\sum_kp_j^{(k)}=P$. 		
	\end{theorem}
	\begin{proof}
		See Appendix~\ref{app_optimal_power_allocation}.
	\end{proof}}
	\begin{remark}
		\label{remark_cmp}
		From (\ref{power_allocation}), we can observe that the power allocation strategy in the RHS-aided semantic communication system differs from the water-filling scheme used in existing RHS-aided bit-level communication networks~\cite{Deng_RHS_multi_user_2022} when the transmit SNR is low and $e_j^{(k)}<1$.
		
		In particular, unlike the water-filling solution, which avoids allocating power to equivalent channels with poor transmission quality (i.e., small channel gain $\sigma_j^{(k)}$), in the semantic communication system, transmit power should be assigned to all equivalent channels to reduce the image reconstruction loss and improve semantic communication performances.
		
	\end{remark}
	{We further provide additional physical insights to clarify why the power allocation strategy here differs from the classical water-filling principle, as discussed in Remark~\ref{remark_cmp}. The key reason lies in the distinct characteristics of the underlying objective functions. Specifically, in the considered semantic communication network, the objective is to minimize the image reconstruction loss through optimized power allocation. It is worth noting that, for each sub-image stream, the reconstruction loss saturates when the transmit power becomes sufficiently large. In contrast, when the allocated power is small, the reconstruction loss increases significantly. Therefore, instead of assigning zero power to equivalent channels with poor transmission quality and concentrating power solely on high-quality channels, allocating a certain amount of power to the weaker channels can effectively reduce the overall reconstruction loss across all data streams. Differently, conventional bit-level communication systems aim to maximize the channel capacity. The capacity of each subchannel increases unboundedly with transmit power. In addition, for subchannels with poor channel conditions, the capacity gain achieved by allocating power to these subchannels is limited. Consequently, to maximize the sum capacity across all subchannels, it is optimal to allocate zero power to low-quality channels and concentrate power on better ones, leading to the classical water-filling solution.}
	
	\subsection{Holographic Beamforming Design}
	{Given the digital beamforming and combining, we try to design a semantic-importance aware holographic beamforming algorithm to solve the subproblem in (\ref{opt_problem_HBF}). Since the the digital beamforming and combining are fixed here, the transmit power of the spatial channels and their allocations are given here. Therefore, we optimize the configurations of the RHS to adjust the channel gains of these spatial channels, so that the received SNR of each subimage can better adapt to its corresponding semantic importance and the over image reconstruction loss can be minimized.}
	
	
	Note that in the digital beamforming subproblem, we have adopted the block-diagonalization-based beamforming framework, as shown in (\ref{digital_BF}) and (\ref{digital_CB}), which can effectively mitigate the inter-user and inter-stream interference. Therefore, we neglect constraints (\ref{no_inter_stream_interference}) and (\ref{no_inter_user_interference}) for the holographic beamforming here. Note that the objective function in (\ref{obj_HBF}) is non-convex with respect to the RHS configuration $\bm{A}$ mainly due to the term $(\gamma_j^{(k)})^{e_j^{(k)}}$ in the denominator. {Therefore, we introduce auxiliary variables $x_j^{(k)}$ to replace the term $(\gamma_j^{(k)})^{e_j^{(k)}}$ and move the term into constraint}, i.e.,
	\begin{subequations}\label{opt_problem_HBF_v2}
		\begin{align}
			\label{obj_HBF_v2}
			&\min_{\bm{A},\{x_j^{(k)}\}_{j,k}} \sum_j\frac{1}{K_j}\sum_{k=1}^{K_j}\left(\frac{b_j^{(k)}}{c_j^{(k)}+x_j^{(k)}}+d_j^{(k)}\right),\\
			\label{cons_max_radiated_power_HBF_v2}
			s.t.~&\Tr(\bm{A}\bm{Q}\bm{W}\bm{W}^{H}\bm{Q}^{H}\bm{A}^{H})\le P,\\
			\label{cons_discrete_polarizability_HBF_v2}
			&a_{n}\in [a_{min},a_{max}],\\
			\label{cons_aux_cons}
			&x_j^{(k)}\le
			 (\gamma_j^{(k)})^{e_j^{(k)}},~\forall j,k,
		\end{align}
	\end{subequations}
	where the objective function in (\ref{obj_HBF_v2}) is now a convex function. However, the newly introduce constraint (\ref{cons_aux_cons}) is still non-convex, making it non-trivial to solve problem (\ref{opt_problem_HBF_v2}).
	
	To address this issue, we employ the successive convex programming~(SCP) method~\cite{Zeng_IOS_2022}, which iteratively solves a sequence of convex optimization subproblems to approach the solution of problem~(\ref{opt_problem_HBF_v2}). In each iteration, the non-convex constraint in~(\ref{cons_aux_cons}) is approximated by a convex function, thereby constructing a tractable convex subproblem. In the following, we classify the constraint in~(\ref{cons_aux_cons}) into two cases depending on the value of the parameter  $e_j^{(k)}$, i.e., $e_j^{(k)} \le 1$ and $e_j^{(k)} > 1$, and explain how to transform the non-convex constraint into a convex one in each case.
	
	
	\subsubsection{$e_j^{(k)}\le 1$}
	First, we rewrite constraint (\ref{cons_aux_cons})~as
	\begin{align}
		\label{cons_aux_cons_v2}
		(x_j^{(k)})^{\frac{1}{e_j^{(k)}}}-\gamma_j^{(k)}\le 0.
	\end{align}
	Since $e_j^{(k)}\le 1$, the first term $(x_j^{(k)})^{\frac{1}{e_j^{(k)}}}$ is a convex function with respect to the auxiliary variable $x_j^{(k)}$. Further, by substituting the expression of digital combiner in (\ref{digital_CB}) and $\bm{E}_j=\bm{D}_j\bm{H}_j\bm{A}\bm{Q}\bm{W}_j$ into (\ref{SNR}), SNR of $k$-th data stream is
	\begin{align}
		\gamma_j^{(k)}=\frac{1}{\sigma^2}|\sum_{n=1}^{N}l_{n}a_{n}|^2,
	\end{align} 
	where $l_{n}$ are constants and $a_{n}$ is the radiation amplitude of the $n$-th RHS element. Therefore, we can find that $\gamma_j^{(k)}$ is convex with respect to the radiation amplitudes of the RHS elements. 
	
	Note that the constraint function in (\ref{cons_aux_cons_v2}) is in the form of a difference of two convex functions. Therefore, in the $o$-th iteration, to transform it into a convex constraint, we replace $\gamma_j^{(k)}$ with its first-order Taylor expansion, which is given by,
	\begin{align}		&\mathcal{T}_o(\gamma_j^{(k)})=\gamma_j^{(k)}\left(\{(a_{n})_{(o-1)}\}\right)\notag\\
&+\sum_{n,n'=1}^{N}\left(2\Re(l_{n'}l_n)(a_{n'})_{(o-1)}\right)\times\left(a_{n}-(a_{n})_{(o-1)}\right),
	\end{align}
	where $\{(a_{n})_{(o-1)}\}$ is the derived radiation amplitudes of the RHS elements in the $(o-1)$-th iteration, and $\Re(\cdot)$ represents the real part of a number. Correspondingly, constraint (\ref{cons_aux_cons_v2}) is transformed into the following convex one,
	\begin{align}
		\label{cons_aux_cons_v3}
		(x_j^{(k)})^{\frac{1}{e_j^{(k)}}}-\mathcal{T}_o(\gamma_j^{(k)})\le 0,
	\end{align}
	
	\subsubsection{$e_j^{(k)}> 1$}
	Similar to the case with parameter $e_j^{(k)}\le 1$, when $e_j^{(k)}> 1$, we can also rearrange the constraint in (\ref{cons_aux_cons}) into the (\ref{cons_aux_cons_v2}). However, unlike the previous case, the first term $(x_j^{(k)})^{\frac{1}{e_j^{(k)}}}$ is now a concave function. Therefore, the overall constraint function in (\ref{cons_aux_cons_v2}) is concave. Similarly, to transform the constraint function in (\ref{cons_aux_cons_v2}) into a convex one, we can utilize its first-order Taylor expansion $\mathcal{T}_o((x_j^{(k)})^{\frac{1}{e_j^{(k)}}}-\gamma_j^{(k)})$.
	Then, in the $o$-th iteration, the non-convex constraint in (\ref{cons_aux_cons_v2}) is transformed into the following convex constraint,
	\begin{align}
		\mathcal{T}_o\left((x_j^{(k)})^{\frac{1}{e_j^{(k)}}}-\gamma_j^{(k)}\right)\le 0.
	\end{align}
	
	Based on the above discussions, we can derive the convex optimization subproblem to be solved in the $o$-th iteration\footnote{Constraint (\ref{cons_max_radiated_power_HBF_v4}) is derived from (\ref{cons_max_radiated_power_HBF_v2}). Specifically, constraint (\ref{cons_max_radiated_power_HBF_v2}) can be rewritten as $\sum_{s=1}^S (\bm{l}_s^H \bm{l}_s)(\bm{W}[s,:]\bm{W}^H[s,:]) \le P$, where $\bm{l}_s$ denotes the $s$-th column of the matrix $\bm{A}\bm{Q}$, and $\bm{W}[s,:]$ represents the $s$-th row of the digital beamformer $\bm{W}$. According to (\ref{cons_max_sum_power_DBF}), we have $\sum_{s=1}^S(\bm{W}[s,:]\bm{W}^H[s,:]) \le P$. Therefore, to satisfy the leakage power constraint, we require that $\bm{l}_s^H \bm{l}_s\le1$. Then, by substituting $\bm{l}_s^H \bm{l}_s=\sum_{n=1}^{N_s}|q_{n,s}|^2a_{(s-1)N_s+n}^2$, where $q_{n,s}$ describes the amplitude and phase changes from the $s$-th feed to the $n$-th RHS element, we can derive the constraint in (\ref{cons_max_radiated_power_HBF_v4}).},
	\begin{subequations}\label{opt_problem_HBF_v4}
		\begin{align}
			\label{obj_HBF_v4_case2}
			&\max_{\bm{A},\{x_j^{(k)}\}_{j,k}} \sum_j\frac{1}{K_j}\sum_{k=1}^{K_j}\left(\frac{b_j^{(k)}}{c_j^{(k)}+x_j^{(k)}}+d_j^{(k)}\right),\\
			\label{cons_discrete_polarizability_HBF_v4_case2}
			s.t.~&a_{n}\in [a_{min},a_{max}],\\
			\label{cons_max_radiated_power_HBF_v4}
			&\sum_{\hat{n}=1}^{N_s}|q_{\hat{n},s}|^2a_{(s-1)N_s+\hat{n}}^2\le 1,\\
			& (x_j^{(k)})^{\frac{1}{e_j^{(k)}}}-\mathcal{T}_o(\gamma_j^{(k)})\le 0, \forall j,k\in\{(j,k)|e_j^{(k)}\le1\},\\
			& \mathcal{T}_o\left((x_j^{(k)})^{\frac{1}{e_j^{(k)}}}-\gamma_j^{(k)}\right)\le 0,\forall j,k\in\{(j,k)|e_j^{(k)}>1\}.
		\end{align}
	\end{subequations}
	The overall holographic beamforming algorithm is summarized in Algorithm~\ref{algorithm_holographic}
	\begin{algorithm}[!tpb]
		\caption{Holographic beamforming design}
		\label{algorithm_holographic} 
		\begin{algorithmic}[1]
			\REQUIRE Digital beamformer $\{\bm{W}_j\}$ and combiner $\{\bm{D}_j\}$
			\STATE Initialize RHS radiation amplitudes $\bm{A}^{(0)}$ and $o=1$
			\REPEAT
			\STATE Optimize RHS radiation amplitudes $\bm{A}$ by solving convex optimization problem (\ref{opt_problem_HBF_v4}), and denote the optimal solution as $\bm{A}^{(o)}$
			\STATE Set $o=o+1$
			\UNTIL{Convergence}
			\ENSURE{The optimal RHS radiation amplitudes $\bm{A}$}
		\end{algorithmic}
	\end{algorithm}

\vspace{-0.2cm}
\subsection{Overall Algorithm Description}
Based on the algorithms presented in the previous two subsections, we design a joint beamforming algorithm to solve problem (\ref{opt_problem}) iteratively. Specifically, in each iteration, the radiation amplitudes $\bm{A}$ of the RHS elements is first optimized through holographic beamforming algorithm in Algorithm~\ref{algorithm_holographic}, given the digital beamformer $\{\bm{W}_j\}_j$ and combiner $\{\bm{D}_j\}_j$. Then, given the holographic beamformer $\bm{A}$, the digital beamformer $\{\bm{W}_j\}_j$ and digital combiner $\{\bm{D}_j\}_j$ are jointly optimized through Algorithm~\ref{algorithm_digital_BF_CB}. The iterations are performed until the value difference of the image reconstruction loss between two adjacent iterations is less than a predefined threshold. The overall algorithm is summarized in Algorithm~\ref{algorithm_all}.

{
\subsection{Extended Discussion}
\subsubsection{Computational complexity}: On the one hand, according to Algorithms~\ref{algorithm_digital_BF_CB}, the complexity of performing SVD on the equivalent channel matrix to eliminate inter-user and inter-stream interference in the proposed digital beamforming scheme is $\mathcal{O}(J^3M^2S + JM^2(S-JM))$. In addition, the complexity of the sub-channel allocation is $\mathcal{O}(\sum_{j} K_j \log K_j)$, and the complexity of the power allocation algorithm is $\mathcal{O}((\sum_{j} K_j)^3)$. Therefore, the total computational complexity of the digital beamforming scheme is $\mathcal{O}(J^3M^2S + JM^2(S-JM) + \sum_{j} K_j \log K_j + (\sum_{j} K_j)^3)$. On the other hand, the holographic beamforming scheme is based on the SCP method. Assuming a total of $I_H$ iterations, where each iteration employs the interior-point method to solve a constrained convex optimization problem with a complexity of $\mathcal{O}((N+JK)^3)$, where $N$ is the number of RHS elements, $J$ is the number of users, and $K$ is the number of data streams transmitted per user. Therefore, the complexity of the holographic beamforming algorithm is $\mathcal{O}(I_H(N+JK)^3)$. Assuming that the digital beamforming and holographic beamforming are solved alternately for $I$ iterations. Then, the total complexity of the proposed semantic-importance-aware holographic beamforming algorithm is $\mathcal{O}(I(J^3M^2S + JM^2(S-JM) + \sum_{j} K_j \log K_j + (\sum_{j} K_j)^3 + I_H(N+JK)^3))$.

\subsubsection{Discussion on beam alignment}
The proposed approach can promote beam alignment between the transmitter and the receivers. Notice that the semantic reconstruction loss is strictly inversely proportional to the receive signal-to-noise ratio (SNR) of each data stream. Therefore, to minimize the reconstruction loss, our optimization algorithm naturally forces the beams to steer toward the users to maximize the transmission SNR of the data streams. Thus, the proposed approach intrinsically facilitates beam alignment. When extended to dynamic scenarios with high user mobility, the beam alignment problems can be addressed by periodically updating instantaneous channel state information and re-executing the proposed algorithm. Specifically, the transmitter can periodically acquire the updated instantaneous channel state information~(CSI) with low overhead by leveraging advanced compressive sensing-based channel estimation methods (e.g.,~\cite{Yue_Hybrid_2025}). Once the newly estimated CSI is obtained, our proposed alternating optimization algorithm can be re-executed to adaptively update the RHS radiation amplitudes along with digital beamformer and combiner. Through this process, the proposed system can effectively track user movements and maintain beam alignment in dynamic environments.

\subsubsection{Comparison with phased array}
To fully capture system-level differences between the RHS and the phased array, we also compare power consumption, resolution, control overhead, and real-time reconfiguration capability between them. Specifically, (a) unlike phased arrays, which rely on power-hungry phase shifters, RHS achieves beamforming through amplitude control using low-power PIN diodes. As a result, the power consumption of RHS is lower than that of phased arrays~\cite{Deng_RHS_multi_user_2022}. (b) Although the resolution of phase shifters in conventional phased arrays is generally higher than that of RHS, existing studies~\cite{Hu_how_many_2022} have shown that, when the number of RHS elements is sufficiently large, the beamforming gain achieved by a 1-bit RHS can approach that of an RHS with continuous amplitude control, thereby avoiding performance loss caused by finite resolution in RHS systems. (c) In terms of control overhead, to achieve high-precision phase control, phased arrays generally require relatively high control overhead for phase shifter operation, e.g., 8-bit control. By contrast, each RHS element only performs 1-bit amplitude control, resulting in lower control overhead. (d) Since RHS performs beamforming using PIN diodes that support high-speed switching, it offers strong real-time reconfiguration capability. For example, in~\cite{Di_RHS_IMT}, the beam switching speed of RHS can reach 2~$\mu$s.

\subsubsection{Discussion on energy efficiency}
Assuming an identical number of antenna elements, a fully digital scheme can achieve higher spectral efficiency compared to the RHS-based hybrid beamforming while the power consumption of the RHS-based scheme is significantly lower, where the overall energy efficiency of the fully digital beamforming scheme is lower than the proposed RHS-based hybrid beamforming scheme. In addition, compared to conventional bit-based communication, semantic communication drastically reduces the number of symbols required for transmission, thereby significantly lowering signal transmission power consumption.
}

\begin{algorithm}[!tpb]
	\caption{Semantic importance aware joint beamforming algorithm}
	\label{algorithm_all}
	\begin{algorithmic}[1] 
		\STATE Initialize digital beamformer $\{\bm{W}_j\}_j$, RHS radiation amplitudes $\bm{A}$, and digital combiner $\{\bm{D}_j\}_j$
		\STATE Set $t=1$		
		\REPEAT
		\STATE Given $\{\bm{W}_j^{(t-1)}\}_j$ and $\{\bm{D}_j^{(t-1)}\}_j$, calculate $\bm{A}^{(t)}$ by solving problem (\ref{opt_problem_HBF}) with Algorithm~\ref{algorithm_holographic};
		\STATE  Given $\bm{A}^{(t)}$, calculate $\{\bm{W}_j^{(t)}\}_j$ and $\{\bm{D}_j^{(t)}\}_j$ by solving problem (\ref{opt_problem_DBF}) using Algorithm~\ref{algorithm_digital_BF_CB};
		\STATE Update $t=t+1$
		\UNTIL{Convergence}		
		\ENSURE Optimal digital beamformer $\{\bm{W}_j^*\}_j$, RHS radiation amplitudes $\bm{A}^*$, and digital combiner $\{\bm{D}_j^*\}_j$.
	\end{algorithmic}
\end{algorithm}
%
%
\vspace{-.2cm}
\section{Simulation Results}
\label{sec_simulation}
In this section, we evaluate the performance of the proposed semantic importance-aware joint beamforming algorithm. Major simulation parameters are set up based on existing works~\cite{Deng_RHS_multi_user_2022,Zeng_both_2021}. Specifically, the center frequency is set to $f=30$~GHz, with the corresponding wavelength given by $\lambda=1$~cm. Noise variance is $\sigma^2=-96$~dBm. The number of RF chains at the BS is $S=8$, where each RF chain is connected to one feed of the RHS. Each feed is connected to one row of $N_{sub}=20$ RHS elements, and thus the RHS consists of $N=160$ elements. The maximum radiation amplitude of one RHS element is set as $a_{max}=\sqrt{0.2}$ and the activation probability of RHS elements is $p_{on}=0.5$. The BS serves $J=2$ users, where it transmits $K_j=4$ data streams simultaneously to each user through spatial multiplexing. The location of the two users are given by $(r,\theta,\phi)=(100~m, 76.8^\circ, -38.7^\circ), (130~m, 77.3^\circ, -13.0^\circ)$, respectively. Channels between the RHS and users are assumed to be Rayleigh faded, with pathloss exponent set as~$4$.


Each image is partitioned uniformly into $K_j=4$ subimages. The semantic importance of each subimage is acquired through the semantic segment anything~(SSA) project~\cite{Chen_SSA_2023}. At the same time, these subimages are processed by the encoder $f_E(\cdot;\bm{\psi}_E)$, which consists of three convolutional layers for source encoding and two ResNet blocks for channel encoding~\cite{Yang_Deep_2022}. At the receiver side, the decoder is composed of two ResNet blocks followed by three convolutional layers.

We leverage the Pascal-VOC 2012~(VOC) dataset~\cite{Everingham_PASCAL_2012}, a widely used benchmark for semantic segmentation with over ten thousand images. Each image pixel in the dataset is assigned one of 21 classes (20 objects plus background). The training split is based on the training set in the VOC data set, where the additive white Gaussian noise~(AWGN) channel is considered, and the SNR is fixed to $10$~dB. The weighted image reconstruction loss in (\ref{weighted_construction_loss}) is selected as the loss function and Adam optimizer with the initiate learning rate of $10^{-4}$ is adopted for training. To isolate the impact of our methodology, we focus on the ``Human" object class, thus setting $C = {15}$ in relation to its definition within the full set of VOC dataset classes. 
In addition, $\theta=0.9$ is the weighting factor for the important pixels. In the test split, the validation set in the VOC dataset is adopted. Further, the $K=4$ subimages generated from each image are then ranked according to their semantic importance. Sub-images with the same rank across different images are grouped together, yielding $K=4$ groups of subimages with distinct average importance levels. Among them, sub-image group $1$ has the highest semantic importance, while group $4$ has the lowest. 

\begin{figure}[!t]
	\centering
	\includegraphics[width=0.4\textwidth]{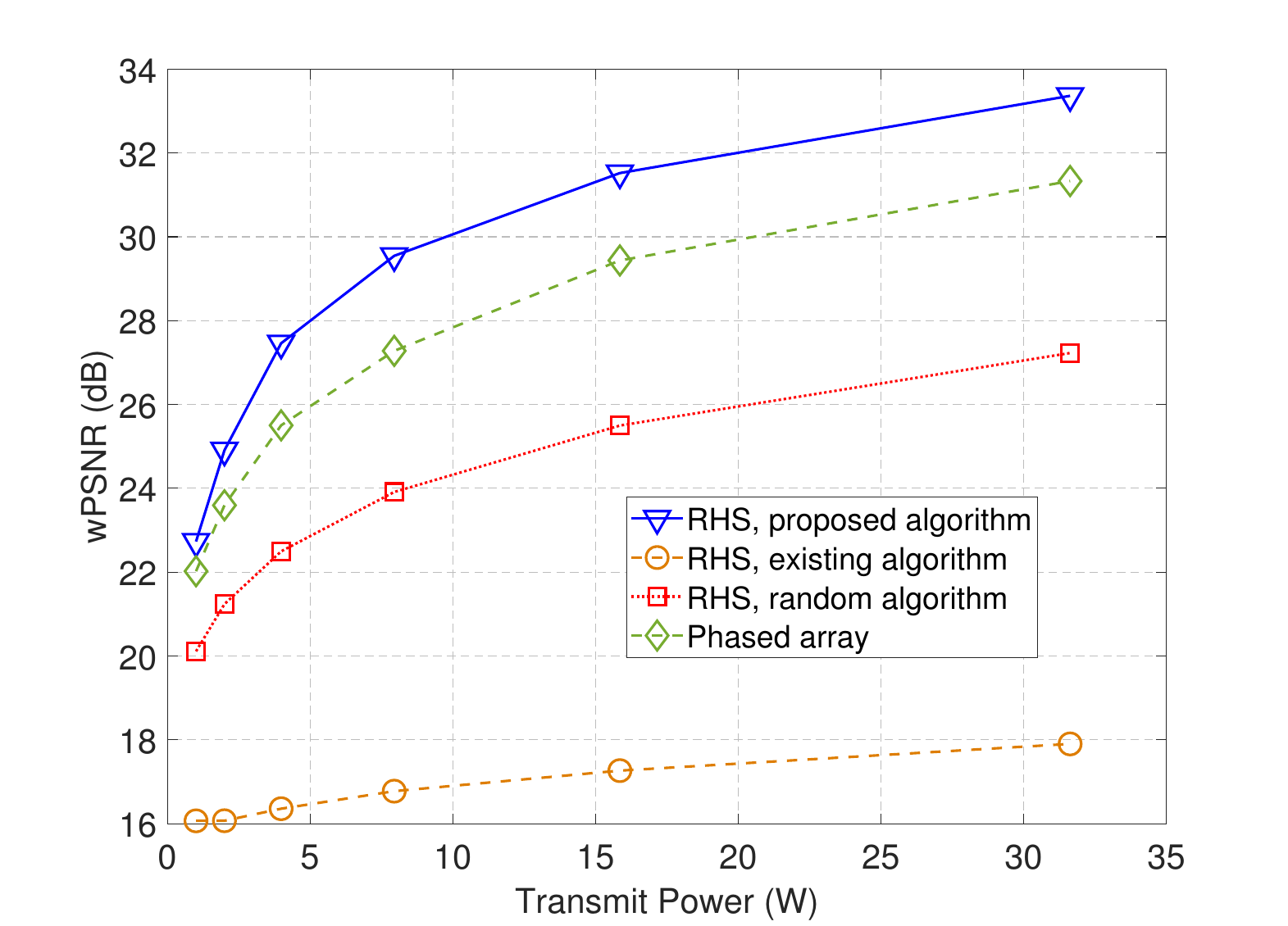}
	\vspace{-1.5mm}
	\caption{wPSNR for image reconstructions versus transmit power with $2$ users. The first three schemes are applied to the considered RHS-aided semantic communication system, where the ``existing algorithm" refers to the holographic beamforming algorithm developed for conventional communication systems. The last one refers to the phased array enabled semantic communication counterpart.}
	\vspace{-5mm}
	\label{wPSNR_vs_SNR}
\end{figure}

\begin{figure}[!t]
	\centering
	\includegraphics[width=0.44\textwidth]{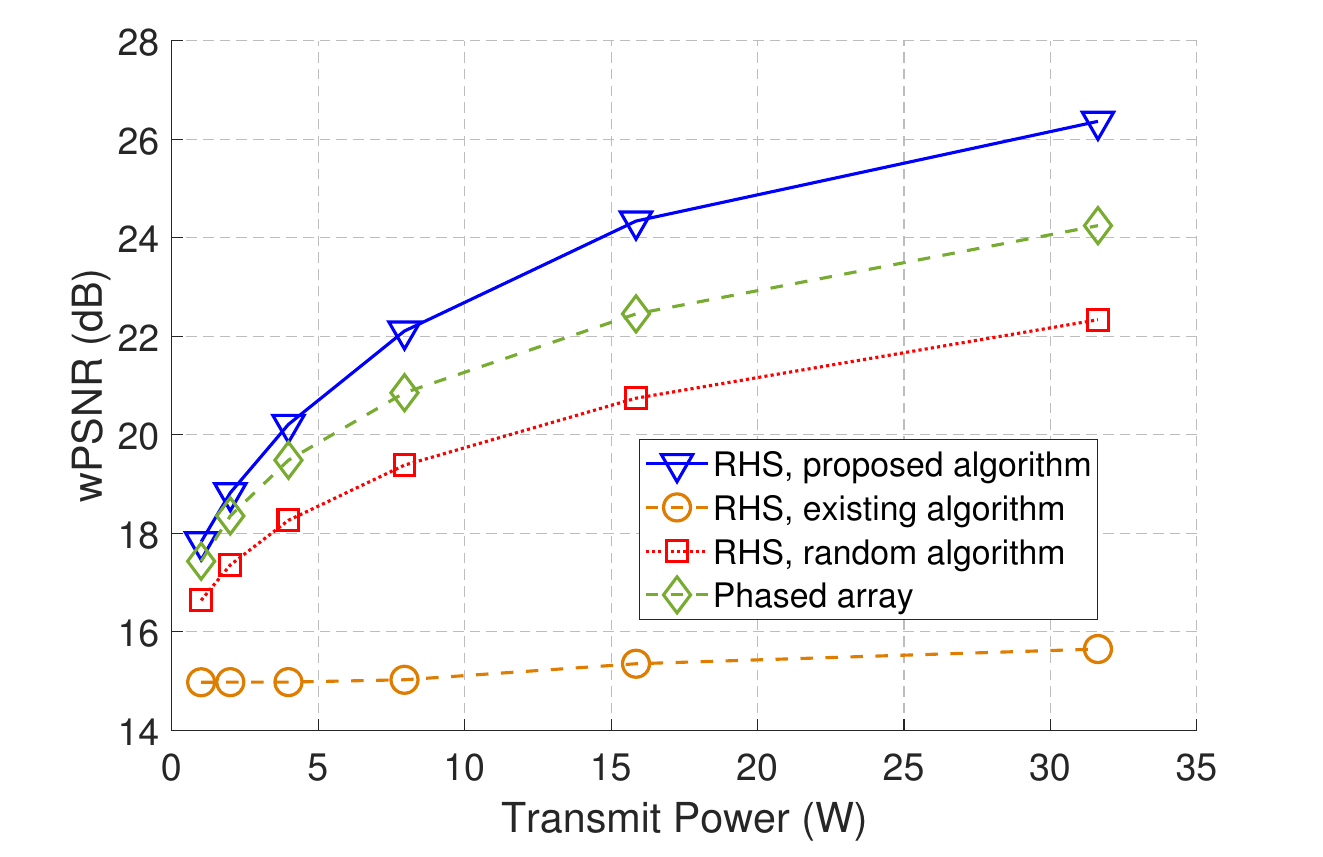}
	\vspace{-1.5mm}
	\caption{wPSNR for image reconstructions versus transmit power with $8$ users. The first three schemes are applied to the considered RHS-aided semantic communication system, where the ``existing algorithm" refers to the holographic beamforming algorithm developed for conventional communication systems. The last one refers to the phased array enabled semantic communication counterpart.}
	\vspace{-5mm}
	\label{wPSNR_vs_SNR_8_user}
\end{figure}

Fig.~\ref{wPSNR_vs_SNR} shows the image reconstruction accuracy of the proposed semantic-aware holographic beamforming algorithm. The image reconstruction accuracy is measured by weighted peak signal to noise ratio~(wPSNR)\footnote{In the simulation results, wPSNR is employed instead of the weighted reconstruction loss because it provides a more intuitive interpretation of reconstruction quality for readers, while being mathematically equivalent to the weighted reconstruction loss.}, {which is defined as the reciprocal of the weighted image reconstruction loss in (\ref{weighted_construction_loss}), i.e., $\frac{1}{\sum_{k=1}^{K_j}\xi_j^{(k)}}$. Such definition has also been adopted in existing works~\cite{Erfurt_study_2019,Zhang_Feature_2026}}. To show the effectiveness of the proposed algorithm, the performances of the three benchmark schemes are also provided for comparison:{
\begin{itemize}
	\item \textbf{RHS-aided random beamforming schemes}: This benchmark adopts the same overall algorithmic framework as our proposed approach. The key distinction lies in the design of the holographic beamforming matrix. Instead of actively optimizing the radiation amplitudes of the RHS elements as detailed in Algorithm~\ref{algorithm_holographic}, this baseline simply selects these amplitudes at random from the feasible set $[a_{min},a_{max}]$. Once the random radiation amplitudes are generated and fixed, the subsequent digital beamforming and combining subproblems are solved utilizing the exact same methodology as our proposed scheme, i.e., Algorithm~\ref{algorithm_digital_BF_CB}. This benchmark serves to demonstrate the specific performance gain brought by our proposed holographic beamforming optimization in Algorithm~\ref{algorithm_holographic}.
	
	\item \textbf{Holographic beamforming schemes developed for conventional communication systems}: While this baseline also uses an alternating optimization framework to update the digital beamformer and combiner along with RHS amplitudes, its subproblems are solved by maximizing the conventional sum rate instead of minimizing the semantic-aware reconstruction loss. Consequently, it ignores the varying semantic importance of different data streams. Because both subproblems are formulated as constrained optimization problems for sum-rate maximization, they are directly solved using the interior-point method via the commercial solver fmincon in MATLAB. Finally, to evaluate the semantic communication performance, the beamforming solutions generated by this sum-rate-driven algorithm are applied to our proposed semantic communication system to compute the resulting wPSNR. This benchmark is introduced to demonstrate that beamforming design in semantic communication systems should explicitly consider semantic importance.
	
	\item \textbf{Phased-array-based beamforming schemes}: Conventional phased arrays are utilized to facilitate semantic MIMO transmissions, where a subconnected architecture is utilized. To ensure a fair comparison, the hardware cost of the phased array is the same as that of an RHS. Similar to our proposed approach, this scheme employs an alternating optimization framework to iteratively update the digital beamforming/combining matrices and the phase-shifter-based analog beamforming matrix. The analog beamforming is performed by optimizing the phase shifters of the phased array to minimize the weighted reconstruction loss through the interior-point method, while the digital beamforming and combining subproblem is solved using the same algorithm in the paper, i.e., Algorithm~\ref{algorithm_digital_BF_CB}. This benchmark is introduced to demonstrate the performance superiority of the proposed RHS-aided semantic communication system over traditional phased arrays..
\end{itemize}}

From Fig.~\ref{wPSNR_vs_SNR}, we can find that the proposed algorithm outperforms the random algorithm, which demonstrates the effectiveness of the proposed holographic beamforming algorithm in Algorithm~\ref{algorithm_holographic}. Further, the proposed algorithm also outperforms the beamforming algorithm developed for conventional bit-communication systems, since the conventional algorithm does not consider the semantic communication performance metric, i.e., image reconstruction loss, leading to performance degradation. In addition, we can find that by applying the proposed beamforming algorithm, the potential of the RHS can be fully uncovered, which outperforms its phased array counterpart. This is because we assume that the hardware cost of the phased array is the same as that of an RHS, and thus the number of RHS elements (i.e., $N=160$ elements) is much larger than that of the phased array elements (i.e., $N=16$ elements), leading to larger radiation aperture and beamforming gain for the RHS\footnote{According to~\cite{Zhang_HISAC_2022}, the number of phased-array elements can be determined based on the number of RHS elements and the ratio of the per-unit manufacturing cost of an RHS element to that of a phased-array element. Here, the ratio of the per-unit cost of RHS elements to phased-array elements is set to $0.1$.}. {We have extended our simulations to a more complex multi-user scenario consisting of $8$ users, as shown in Fig.~\ref{wPSNR_vs_SNR_8_user}. As shown in Fig.~\ref{wPSNR_vs_SNR_8_user}, our proposed method still outperforms the baseline schemes, which demonstrates the effectiveness of the proposed method.}

{
\begin{figure}[!t]
	\centering
	\includegraphics[width=0.43\textwidth]{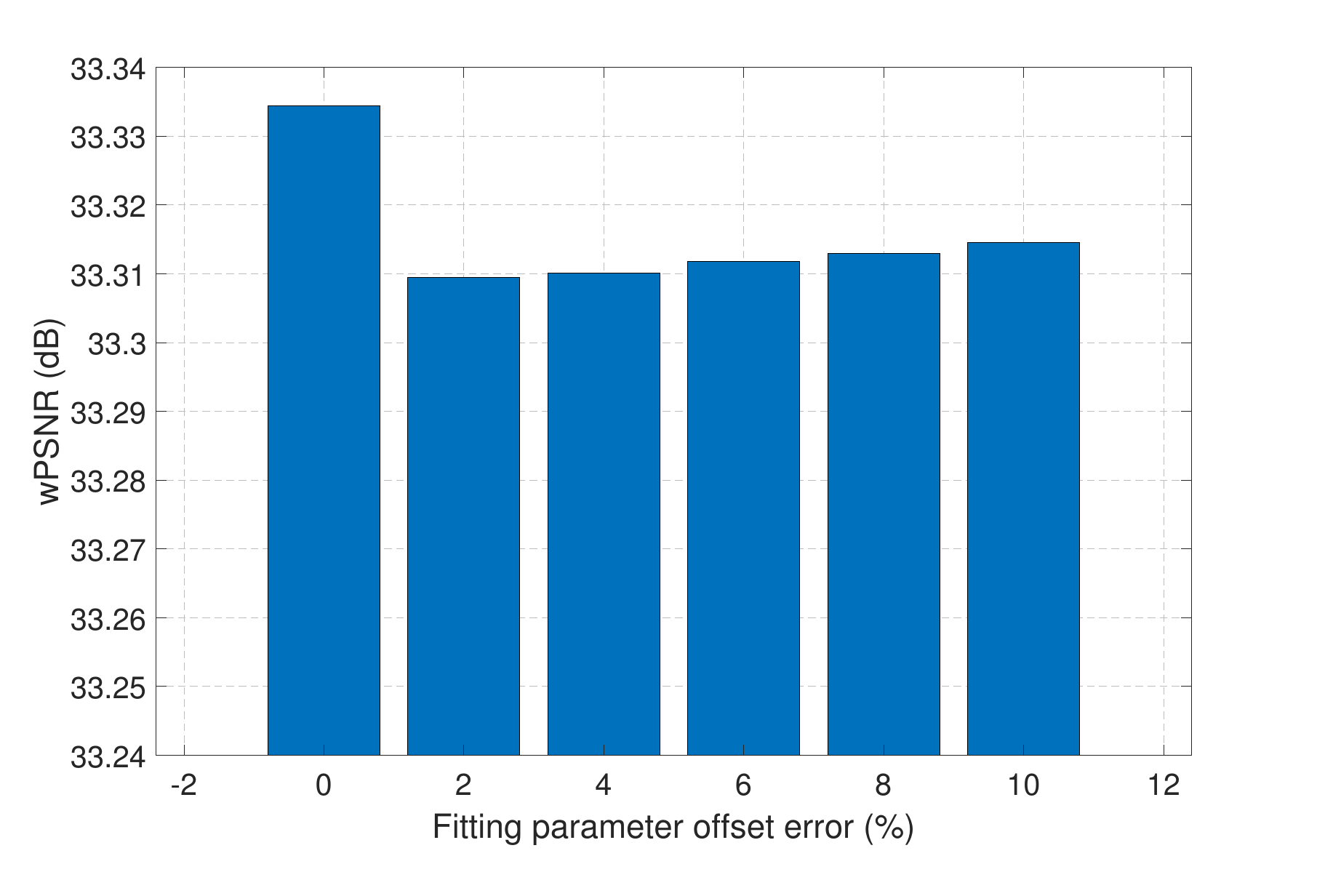}
	\vspace{-1.5mm}
	\caption{Impact of fitting parameter offset errors on the image reconstruction performance of the proposed semantic-aware holographic beamforming scheme.}
	\vspace{-1mm}
	\label{fitting_error}
\end{figure}}

{
Fig.~\ref{fitting_error} presents a robustness study by introducing error offsets to the fitted curve parameters. Specifically, the fitting parameters used in the curve-fitting model are perturbed by different offset errors and then redesign the beamforming based on the perturbed parameters. Then, we evaluate the performance of the beamforming solution based on the real fitting parameters without offset. The simulation results show that the proposed beamforming design is not sensitive to such fitting errors, i.e., when the fitting-parameter offset error increases up to \(10\%\), the achieved wPSNR reduces only from about \(33.33\)~dB to \(33.31\)~dB. This demonstrates that the proposed method is robust to fitting errors.}

\begin{figure}[!t]
	\centering
	\includegraphics[width=0.39\textwidth]{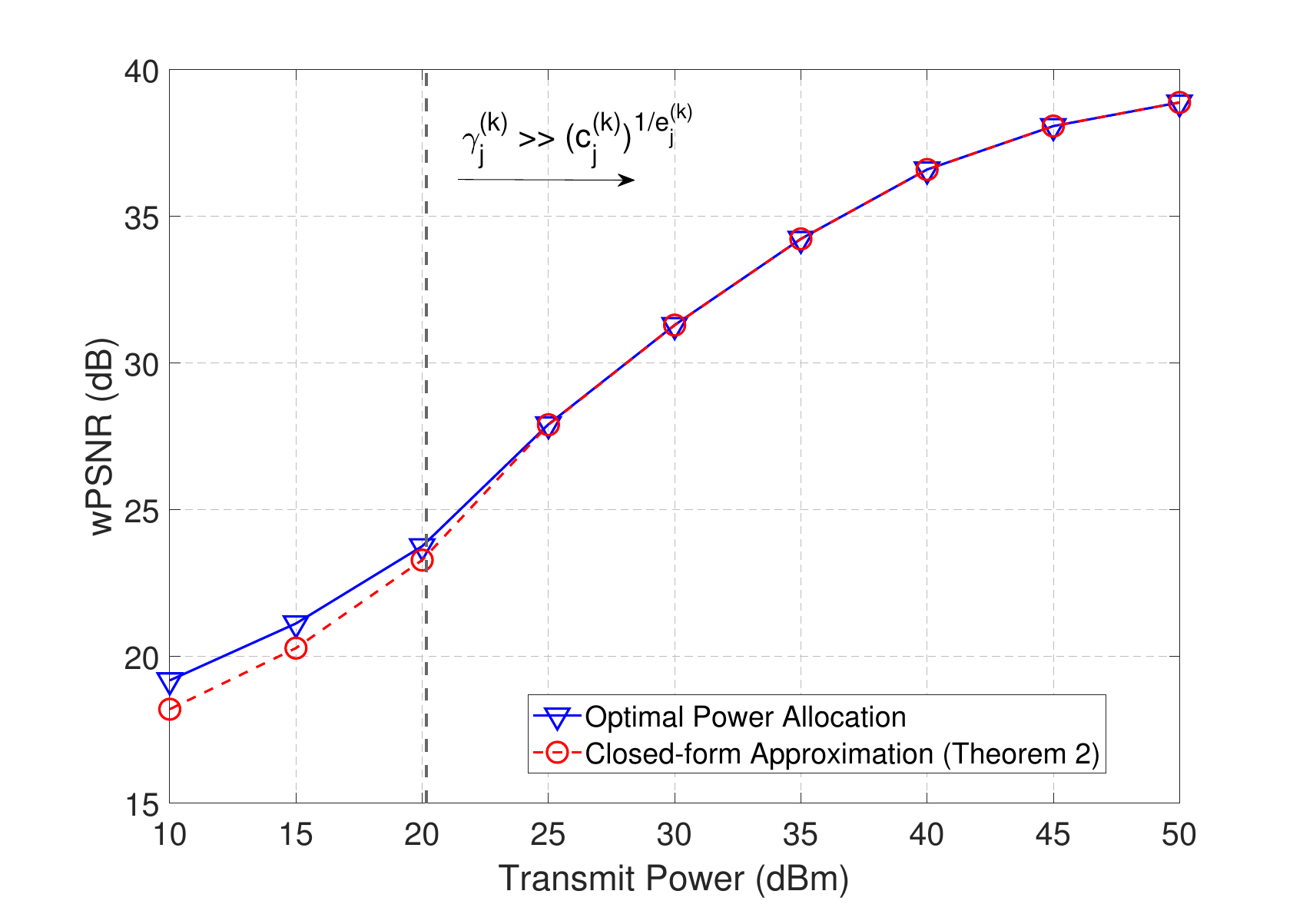}
	\vspace{-1.5mm}
	\caption{Comparison of image reconstruction performance between the closed-form approximation derived in Theorem~2 and the optimal power allocation scheme. The region on the right of the dotted vertical line corresponds to the received SNR condition satisfying $\gamma_j^{(k)}\gg (c_j^{(k)})^{\frac{1}{e_j^{(k)}}}$, i.e., $\gamma_j^{(k)}\ge 10\times(c_j^{(k)})^{\frac{1}{e_j^{(k)}}}$.}
	\vspace{-3mm}
	\label{wPSNR_cmp}
\end{figure}

{
To verify Theorem~\ref{the_optimal_power_allocation} and to demonstrate the performance loss outside the high SNR regime, we compare the image reconstruction performance between the closed-form approximation derived in Theorem~\ref{the_optimal_power_allocation} and the optimal power allocation scheme, as shown in Fig.~\ref{wPSNR_cmp}. Here, the transmit power indicated by the dotted vertical line corresponds to an average received SNR $\gamma_j^{(k)}= 10\times(c_j^{(k)})^{\frac{1}{e_j^{(k)}}}$, and thus the right regime of the dotted vertical line corresponds to the received SNR condition satisfying $\gamma_j^{(k)}\gg (c_j^{(k)})^{\frac{1}{e_j^{(k)}}}$, where we could find that the system performance achieved by the closed-form power allocation solution in Theorem~\ref{the_optimal_power_allocation} approaches that under optimal power allocation, thus verifying  Theorem~\ref{the_optimal_power_allocation}. In addition, we could find that when the received SNR condition does not hold, the performance gap induced by the approximation in Theorem~\ref{the_optimal_power_allocation} is still acceptable. For example, when the transmit power (and thus the received SNR) is $10$~dB lower than the derived threshold, the system performance loss is less than $5\%$.}
\begin{figure}[!tpb]
	\centering
	
	\subfigure[Channel allocations in semantic communication systems, where the equivalent channels with higher gain are allocated to sub-image groups with higher semantic importance.]{
		\begin{minipage}[b]{0.21\textwidth}
			\centering
			\includegraphics[width=1\textwidth]{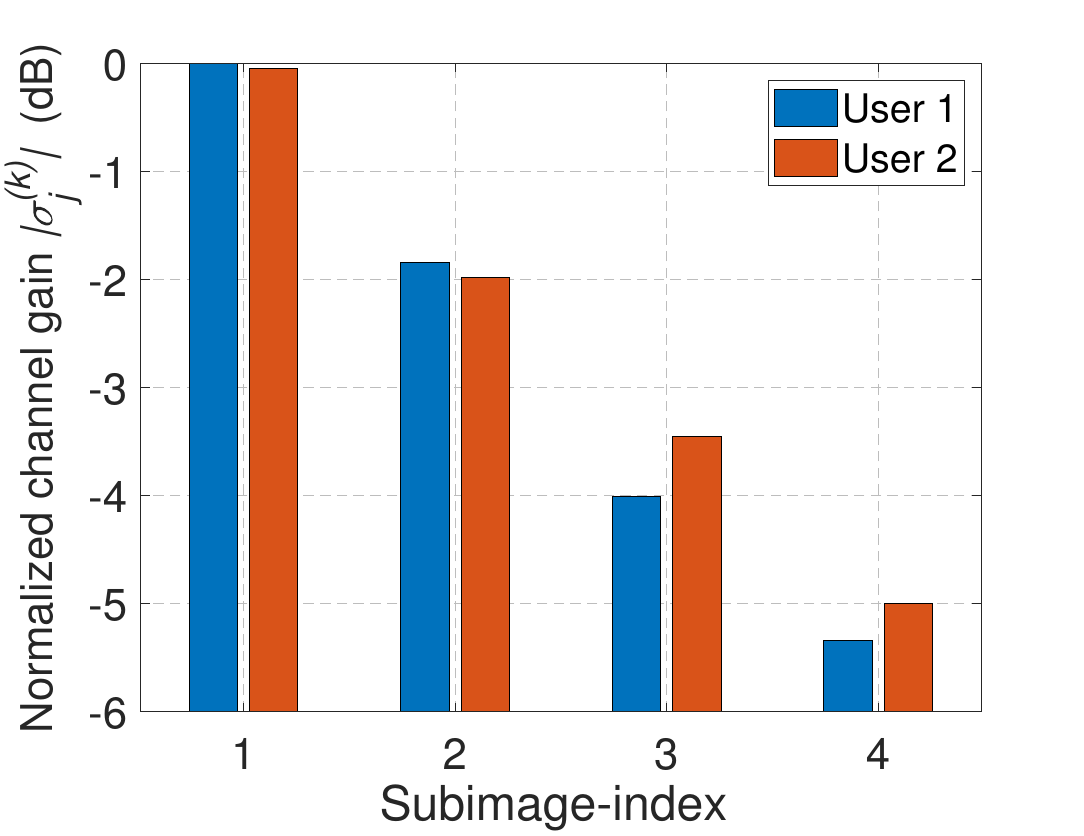}
						\vspace{-0.6cm}
			\label{fig_channel_allocation_semcom}
	\end{minipage}}
	\hspace{0.02\textwidth}
	\subfigure[Channel allocations in conventional communication systems, where the channels are randomly allocated to different sub-images regardless of their semantic importance.]{
		\begin{minipage}[b]{0.21\textwidth}
			\centering
			\vspace{-0.4cm}
			\includegraphics[width=1\textwidth]{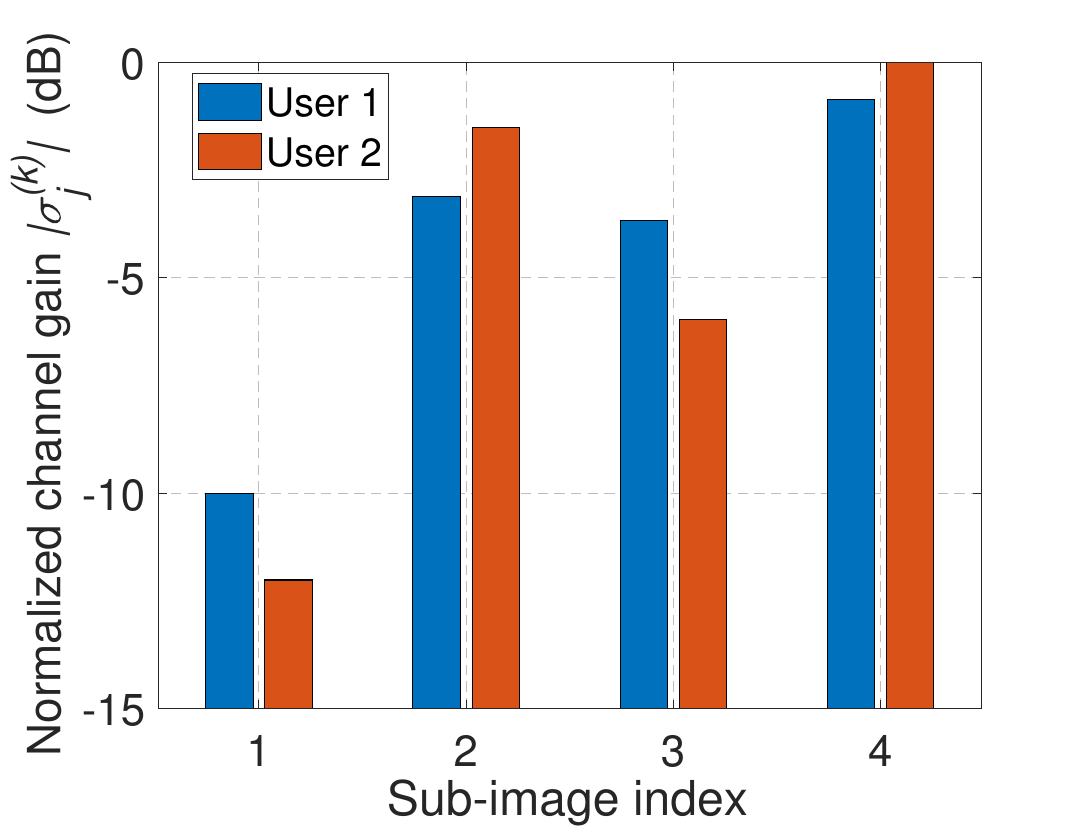}
						\vspace{-0.6cm}
			\label{fig_channel_allocation_convcom}
	\end{minipage}}
	
	\vspace{-0.1cm}
	\subfigure[Power allocations in semantic communication systems, where the channels with low gains should also be  allocated transmit power.]{
		\begin{minipage}[b]{0.21\textwidth}
			\centering
			\includegraphics[width=1\textwidth]{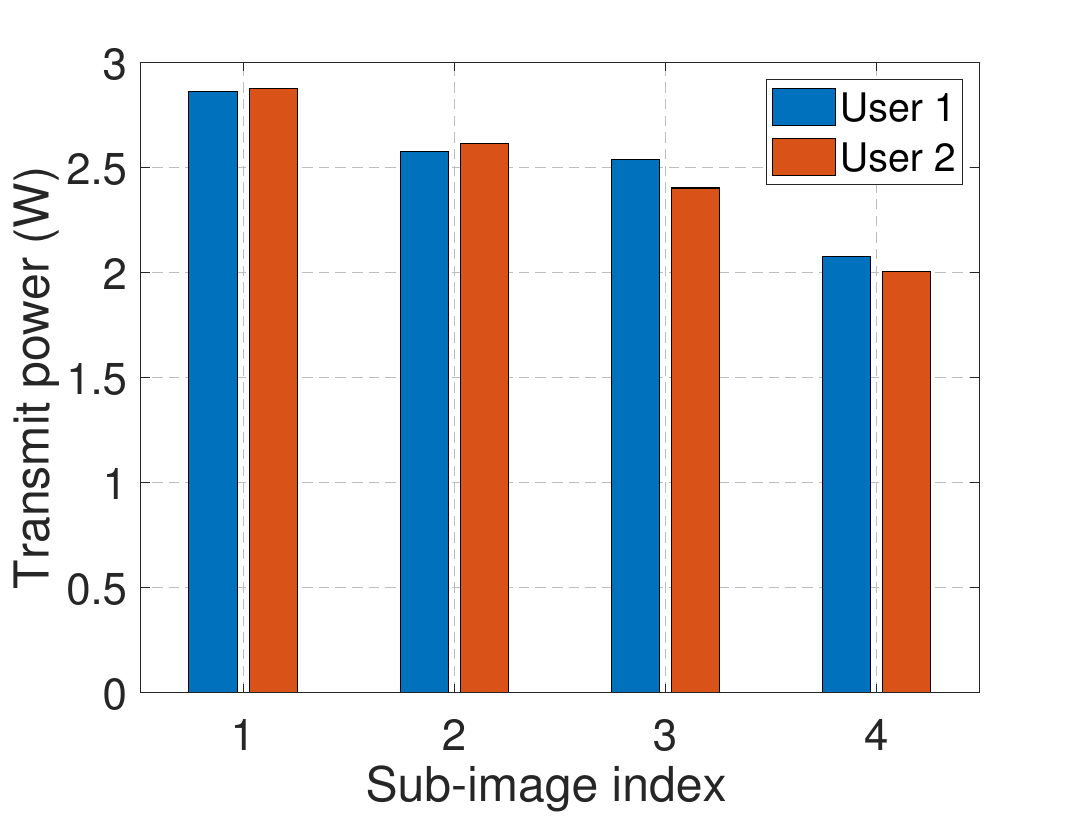}
						\vspace{-0.6cm}
			\label{fig_power_allocation_semcom}
	\end{minipage}}
	\hspace{0.02\textwidth}
	\subfigure[Power allocations in conventional communication systems, where the water-filling scheme is utilized, and thus transmit power will not be allocated to channels of low gains.]{
		\begin{minipage}[b]{0.21\textwidth}
			\centering
			\vspace{-0.4cm}
			\includegraphics[width=1\textwidth]{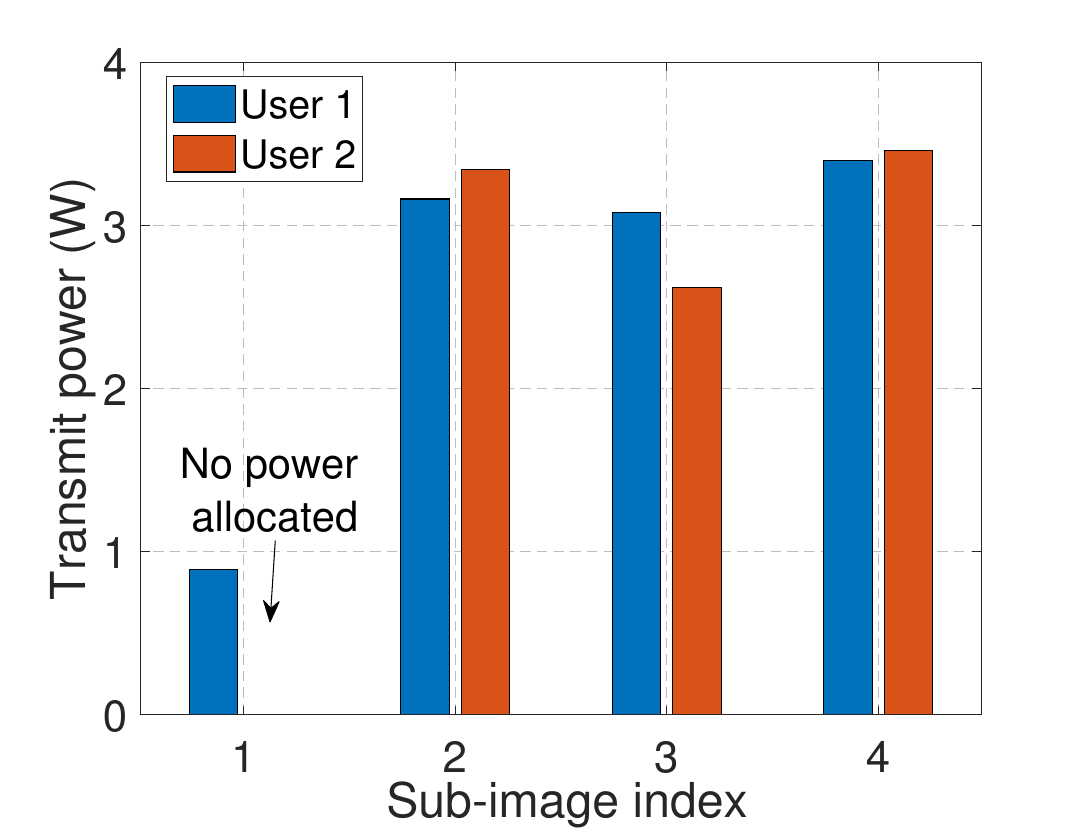}
						\vspace{-0.6cm}
			\label{fig_power_allocation_convcom}
	\end{minipage}}
		\vspace{-0.2cm}
	\caption{Comparison of channel allocation and power allocation schemes between semantic communication systems and conventional bit-communication systems. The semantic importance of the subimage groups decreases with their index, i.e., subimage group $1$ has the highest importance while $4$ has the lowest importance. In contrast, in conventional bit-communication, semantic importance is ignored.}
	\label{fig_power_allocation_comparison}
\end{figure}

Fig.~\ref{fig_power_allocation_comparison} compares the channel and power allocation schemes between the considered semantic communication systems and the conventional communication systems. The transmit power is set as $45$~dBm. The semantic importance of the subimage group decreases with its index, i.e., subimage group $1$ and $4$ has the highest and lowest importance. From Fig.~\ref{fig_channel_allocation_semcom}, we can find that in the semantic communication system, equivalent channels with higher gains $|\sigma_j^{(k)}|$ are allocated to data streams with higher semantic importance. However, for conventional communication systems, the semantic importance is neglected during the channel allocations, as shown in Fig.~\ref{fig_channel_allocation_convcom}. Further, by comparing Fig.~\ref{fig_power_allocation_semcom} against Fig.~\ref{fig_power_allocation_semcom}, we can find that unlike the conventional communication systems, which adopt the water-filling power allocation scheme, and thus the equivalent channels with low gains may not be allocated any transmit power, in the considered semantic communication system, each equivalent channel is allocated transmit power. This result is consistent with Remark~\ref{remark_cmp}\footnote{Simulation results show that the fitting curve parameters $e_j^{(k)}<1, \forall j,k$. Further, the received SNR of different data streams is low, ranging from $-9$~dB to $-1$~dB.}.

\begin{figure*}[!t]
	\centering
	\includegraphics[width=0.9\textwidth]{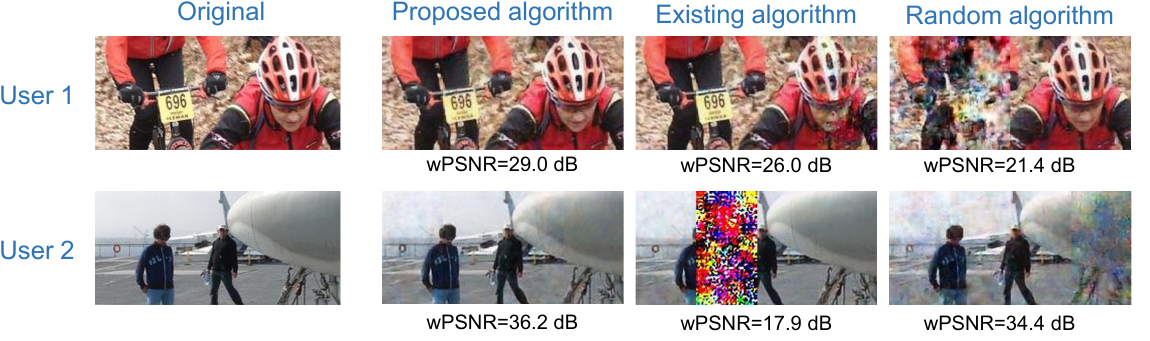}
	\caption{Examples of the original image and the reconstructed images by using different schemes.}
	\vspace{-4mm}
	\label{figure_dif_scheme}
\end{figure*}

To visualize the performance of different schemes, an example is provided in Fig.~\ref{figure_dif_scheme}. As observed, the proposed holographic beamforming achieves the highest image recovery accuracy. In contrast, for the existing holographic beamforming algorithm designed for conventional communication systems, the recovered image at user~1 shows distortion in the most critical region, i.e., the human face, since the semantic importance of different image components is not considered. Moreover, for the recovered image at user~2, one subimage containing the human face and body is completely distorted. This is because the conventional scheme employs water-filling power allocation, where channels with poor quality receive little or no transmit power. In addition, the random algorithm yields lower overall image quality than the proposed method, as the RHS is randomly configured and its beamforming potential cannot be fully exploited, resulting in degraded received SNR.


\vspace{-.4cm}

\section{Conclusion}
\vspace{-.1cm}
\label{sec_conclusion}


In this paper, we have studied an RHS-aided multi-user semantic communication system for image transmission. To reduce latency, each image has been divided into multiple sub-images for parallel transmissions, where semantic information has been extracted from each sub-image and its semantic importance has been evaluated through a semantic-segmentation-based method. To enable the simultaneous transmissions of multi-stream semantic information to multiple users, we have proposed a hybrid beamforming framework that combines digital beamforming at the BS, holographic beamforming at the RHS, and digital combining at the users. A tractable approximation of the image reconstruction loss has been derived, revealing the joint impact of semantic importance and received SNR. Based on this, we have formulated and solved a semantic-aware holographic beamforming problem through a semantic-importance-aware method. Based on our analysis and numerical results, the following conclusions can be drawn:
\begin{itemize}
	\item Unlike conventional RHS-aided bit-communications, where the equivalent spatial channels are randomly allocated among different transmit bit sequences, we theoretically prove that the equivalent channels with high channel gains should be allocated for transmitting semantic information of higher importance in RHS-aided semantic communication systems.
	
	\item A closed-form power allocation scheme is derived for the low SNR regime. The derived scheme shows that, unlike the water-filling power allocation scheme developed for conventional bit-communication systems, where the transmit bit sequences experiencing poor channel conditions will not be allocated transmit power in low SNR regime, all transmitted semantic information should be allocated transmit power in semantic communication systems.
	\item Through the proposed semantic-aware holographic beamforming scheme, the RHS-aided semantic communication system outperforms conventional phased-array-enabled semantic MIMO communication systems given the same hardware costs.
\end{itemize}

\begin{appendices}
	\section{Definitions of Matrices Related to Digital Beamformer and Combiner in (\ref{digital_BF}) and (\ref{digital_CB})}
	\label{app_def_db_dc}
	In this section, we introduce the definitions of the constant matrices $\widetilde{\bm{V}}_j^{(0)}$, $\bm{V}_j^{(1)}$, and $\bm{U}_j^{(1)}$ present in block diagonalization-based beamforming framework in (\ref{digital_BF}), (\ref{digital_CB}).
	\subsection{Interference Channel Matrix and Its Null Space}
	For user $j$, define the interference channel matrix $\widetilde{\bm{H}}_j$ as the the combination of effective channels of all users except user $j$, i.e., $\widetilde{\bm{H}}_j=\big[(\bm{H}_{1}\bm{A}\bm{Q})^T,\dots,(\bm{H}_{j-1}\bm{A}\bm{Q})^T,(\bm{H}_{j+1}\bm{A}\bm{Q})^T,\dots,(\bm{H}_{J}\bm{A}\bm{Q})^T\big]^T$.
	Consider the singular value decomposition (SVD) of $\widetilde{\bm{H}}_j$, i.e., $\widetilde{\bm{H}}_j=\widetilde{\bm{U}}_j\widetilde{\bm{\Sigma}}_j\widetilde{\bm{V}}_j$,
	and denote $\widetilde{L}_j=\text{rank}(\widetilde{\bm{H}}_j)$. Partition $\widetilde{\bm{V}}_j=[\widetilde{\bm{V}}_j^{(1)},\widetilde{\bm{V}}_j^{(0)}]$, where $\widetilde{\bm{V}}_j^{(0)}$ spans the null space of $\widetilde{\bm{H}}_j$. This null space basis can be used to suppress inter-user interference.
	\vspace{-.3cm}
	\subsection{Equivalent Channel and SVD-Based Decomposition}
	\label{app_sub_SVD}
	Consider the equivalent channel from the BS to user $j$ after projection onto the null space, i.e., $\bm{H}_j\bm{A}\bm{Q}\widetilde{\bm{V}}_j^{(0)}$.
	Let $\bar{L}_j=\text{rank}(\bm{H}_j\bm{A}\bm{Q}\widetilde{\bm{V}}_j^{(0)})$, and write its SVD as $\bm{H}_j\bm{A}\bm{Q}\widetilde{\bm{V}}_j^{(0)}=\bm{U}_j^{'}\bm{\Sigma}_j^{'}(\bm{V}_j^{'})^H$,
	where $\bm{U}_j^{'}$ and $\bm{V}_j^{'}$ are unitary matrices, and $\bm{\Sigma}_j^{'}$ is diagonal with singular values $\{\sigma_j^{(k)}\}$. Denote by $\bm{U}_j^{(1)}$ and $\bm{V}_j^{(1)}$ the first $\bar{L}_j$ columns of $\bm{U}_j^{'}$ and $\bm{V}_j^{'}$, respectively, which can be used to eliminate inter-stream interference and realize spatial multiplexing.
	\section{Proof of Theorem~\ref{the_optimal_channel_allocation}}
	\label{app_optimal_channel_allocation}
	We first show that under optimal channel and power allocations, subimages with higher semantic importance attain higher received SNR.  
	
	Consider user $j$ and two subimages $k_1$ and $k_2$, where $k_1$ is more important than $k_2$. Suppose under allocations $\{\bm{O}_j\}_j$ and $\{\bm{P}_j\}_j$, we have $\gamma_j^{(k_1)} < \gamma_j^{(k_2)}$.
	The reconstruction loss is
	\begin{align}
		\label{loss_original}
		\xi_j^{(k_1)}+\xi_j^{(k_2)}=f_j^{(k_1)}(\gamma_j^{(k_1)})+f_j^{(k_2)}(\gamma_j^{(k_2)}).
	\end{align}
	If we exchange the channel and power allocations between $k_1$ and $k_2$, the new loss becomes
	\begin{align}
		\label{loss_new}
		\hat{\xi}_j^{(k_1)}+\hat{\xi}_j^{(k_2)}=f_j^{(k_1)}(\gamma_j^{(k_2)})+f_j^{(k_2)}(\gamma_j^{(k_1)}).
	\end{align}
	Their difference is
	\begin{align}
		\label{loss_diff}
		&(\xi_j^{(k_1)}+\xi_j^{(k_2)})-(\hat{\xi}_j^{(k_1)}+\hat{\xi}_j^{(k_2)}) \notag\\
		&=(f_j^{(k_1)}(\gamma_j^{(k_1)})\!-\! f_j^{(k_1)}(\gamma_j^{(k_2)}))\!-\!(f_j^{(k_2)}(\gamma_j^{(k_1)})\!-\!f_j^{(k_2)}(\gamma_j^{(k_2)})).
	\end{align}
	Since when the received SNR changes, more important subimages induce larger changes in reconstruction loss (see Figs.~\ref{reconstruction_loss_vs_SNR}), we have
	\begin{align}
		\label{larger_change_loss}
		|f_j^{(k_1)}(\gamma_j^{(k_1)})\!-\!f_j^{(k_1)}(\gamma_j^{(k_2)})|\! >\! |f_j^{(k_2)}(\gamma_j^{(k_1)})\!-\! f_j^{(k_2)}(\gamma_j^{(k_2)})|.
	\end{align}
	Further, note that reconstruction loss decreases with SNR, i.e.,
	\begin{align}
		\label{decrease_with_SNR}
		f_j^{(k_1)}(\gamma_j^{(k_1)})\!>\! f_j^{(k_1)}(\gamma_j^{(k_2)}),\quad
		f_j^{(k_2)}(\gamma_j^{(k_1)})\!>\! f_j^{(k_2)}(\gamma_j^{(k_2)}).
	\end{align}
	By substituting (\ref{larger_change_loss}) and (\ref{decrease_with_SNR}) into (\ref{loss_diff}), we can find that the swapped allocation yields a smaller total loss. Hence, subimages with higher importance should receive higher SNR.

	Next, we show that higher-gain channels should be assigned to more important subimages. Suppose subimage $k_1$ is more important, but channel $\hat{k}_1$ assigned to $k_1$ has lower gain than $\hat{k}_2$ assigned to $k_2$, i.e.,
	\begin{align}
		\sigma_j^{(\hat{k}_1)} < \sigma_j^{(\hat{k}_2)}.
	\end{align}
	Let $p_j^{(k_1)}$ and $p_j^{(k_2)}$ be powers allocated to subimage $k_1$ and $k_2$, respectively. Then, we swap the channels allocated to subimages $k_1$ and $k_2$. To maintain the same SNR after swapping channels, the required powers are
	\begin{align}
		\widetilde{p}_j^{(k_1)}=\tfrac{\gamma_j^{(k_1)}\sigma^2}{(\sigma_j^{(\hat{k}_2)})^2},\quad
		\widetilde{p}_j^{(k_2)}=\tfrac{\gamma_j^{(k_2)}\sigma^2}{(\sigma_j^{(\hat{k}_1)})^2}.
	\end{align}
	The total power change is $(\widetilde{p}_j^{(k_1)}+\widetilde{p}_j^{(k_2)})-(p_j^{(k_1)}+p_j^{(k_2)})=\frac{\sigma^2(\gamma_j^{(k_1)}-\gamma_j^{(k_2)})((\sigma_j^{(\hat{k}_1)})^2-(\sigma_j^{(\hat{k}_2)})^2)}{(\sigma_j^{(\hat{k}_1)}\sigma_j^{(\hat{k}_2)})^2}$.
	Since $\gamma_j^{(k_1)}>\gamma_j^{(k_2)}$ and $\sigma_j^{(\hat{k}_1)}<\sigma_j^{(\hat{k}_2)}$, total power change is negative, i.e., less power is required after swapping. In addition, saved power can be reallocated to further improve performance, proving that better channels should be assigned to more important~subimages.

	\vspace{0.3cm}
{
	\section{Proof of Theorem~\ref{the_optimal_power_allocation}}
	\label{app_optimal_power_allocation}
Since the power allocation problem in (\ref{opt_problem_DBF_v3}) satisfies Slater’s condition, the optimal solution must satisfy the KKT conditions, which enable closed-form derivations. Specifically, the Lagrangian function is
\begin{align}
	\label{lag}
	&\mathcal{L}(\{p_j^{(k)}\},\{\lambda_j^{(k)}\},\mu)\notag\\
	=&\sum_j\frac{1}{K_j}\sum_{k=1}^{K_j}\!\left(\frac{b_j^{(k)}}{c_j^{(k)}+\big(\tfrac{(\sigma_j^{(o(k))})^2p_j^{(k)}}{\sigma^2}\big)^{e_j^{(k)}}}+d_j^{(k)}\right)\notag\\
	&-\sum_j\sum_{k=1}^{K_j}\lambda_j^{(k)}p_j^{(k)}+\mu\!\left(\sum_j\sum_{k=1}^{K_j}p_j^{(k)}-P\right).
\end{align}
From the KKT conditions, optimal powers $\{(p_j^{(k)})^*\}$ satisfy
\begin{align}
	\label{stationary}
	\frac{\partial\mathcal{L}}{\partial p_j^{(k)}}=0, \ \forall j,k, \qquad
	-\lambda_j^{(k)}p_j^{(k)}=0, \ \forall j,k,
\end{align}
which correspond to stationarity and complementary slackness.  

In the high SNR regime (i.e., the active operating regime where $(\gamma_j^{(k)})^{e_j^{(k)}} \gg c_j^{(k)}$), the constant $c_j^{(k)}$ in the denominator in (\ref{lag}) can be omitted, and objective term is approximated as:
\begin{align}
	&\frac{b_j^{(k)}}{c_j^{(k)}+\!\left(\tfrac{(\sigma_j^{(o(k))})^2p_j^{(k)}}{\sigma^2}\right)^{e_j^{(k)}}}+d_j^{(k)}\notag\\
	&\approx \frac{b_j^{(k)}}{\!\left(\tfrac{(\sigma_j^{(o(k))})^2p_j^{(k)}}{\sigma^2}\right)^{e_j^{(k)}}}+d_j^{(k)}.
\end{align}
Thus, the stationarity condition becomes
\begin{align}
	\label{stationary_v2}
	-\frac{b_j^{(k)}e_j^{(k)}}{\big(\tfrac{(\sigma_j^{(o(k))})^2}{\sigma^2}\big)^{e_j^{(k)}}}(p_j^{(k)})^{-e_j^{(k)}-1}+\mu-\lambda_j^{(k)}=0.
\end{align}

Then, we will show $\lambda_j^{(k)}=0$. To achieve this, we first prove that each subimage should be allocated transmit power under the optimal transmit power allocations. Suppose a subimage $k_1$ is not allocated power, i.e., $p_j^{(k_1)}=0$, while another subimage $k_2$ is allocated $p_j^{(k_2)}=p_0>0$. Then, let $\Delta$ be a small quantity. Consider the case where the transmit power for subimage $k_1$ increases from $p_j^{(k_1)}=0$ to $p_j^{(k_1)}=\Delta$, leading to an decrease in image reconstruction loss given by
\begin{align}
	\Delta \xi_j^{(k_1)}=-\frac{\big(\tfrac{(\sigma_j^{(o(k_1))})^2}{\sigma^2}\big)^{e_j^{(k_1)}} b_j^{(k_1)} \Delta^{e_j^{(k_1)}}}{\big(c_j^{(k_1)}+(\tfrac{(\sigma_j^{(o(k_1))})^2\Delta}{\sigma^2})^{e_j^{(k_1)}}\big)c_j^{(k_1)}} < 0.
\end{align}
Since $e_j^{(k_1)}<1$, we have
\begin{align}
	\label{lim_k_1}
	\lim_{\Delta\to 0^+}\Big|\frac{\Delta \xi_j^{(k_1)}}{\Delta}\Big| = +\infty,
\end{align}
showing that a tiny increment $\Delta$ yields a large accuracy gain.  

Correspondingly, when transmit power $p_j^{(k_2)}$ for the $k_2$-th subimage decreases from $p_0$ to $p_0-\Delta$, the image reconstruction loss becomes more severe, with changes given by
\begin{align}
	\Delta \xi_j^{(k_2)}=\frac{b_j^{(k_2)}e_j^{(k_2)}\big(\tfrac{(\sigma_j^{(o(k_2))})^2}{\sigma^2}\big)^{e_j^{(k_2)}}p_0^{\,e_j^{(k_2)}-1}}{\big(c_j^{(k_2)}+(\tfrac{(\sigma_j^{(o(k_2))})^2p_0}{\sigma^2})^{e_j^{(k_2)}}\big)^2}\Delta + o(\Delta) > 0,
\end{align}
based on which we can find that
\begin{align}
	\label{lim_k_2}
	\lim_{\Delta\to 0^+}\Big|\frac{\Delta \xi_j^{(k_2)}}{\Delta}\Big| < +\infty.
\end{align}
Based on (\ref{lim_k_1}) and (\ref{lim_k_2}), we can conclude that transferring $\Delta$ transmit power from the $k_2$-the subimage to the $k_1$-th subimage, i.e., $k_1$ is allocated transmit power, can improve the overall image reconstruction accuracy. Therefore, under the optimal transmit power allocations, each subimage should be allocated transmit power, i.e.,
\begin{align}
	\label{power_positive}
	p_j^{(k)} > 0,
\end{align}
Based on (\ref{power_positive}) and (\ref{stationary}), we can derive that $\lambda_j^{(k)}$ satisfy 
\begin{align} 
	\label{Lag_multiplier_zero} \lambda_j^{(k)}=0, 
\end{align} 
By substituting (\ref{Lag_multiplier_zero}) into (\ref{stationary_v2}), we have 
\begin{align} 
	-\frac{b_j^{(k)}e_j^{(k)}}{\big(\tfrac{(\sigma_j^{(o(k))})^2}{\sigma^2}\big)^{e_j^{(k)}}}(p_j^{(k)})^{-e_j^{(k)}-1}+\mu=0, 
\end{align} 
based on which we derive optimal power allocation in (\ref{power_allocation}).}

\end{appendices}

\end{document}